\documentclass[10pt,journal,compsoc]{IEEEtran}
\usepackage{amsmath,amssymb,amsfonts}
\usepackage{algorithmicx}
\usepackage[noend,ruled, linesnumbered, vlined]{algorithm2e}
\usepackage{graphicx}
\usepackage{epstopdf}
\usepackage{textcomp}
\usepackage{xcolor}
\usepackage{bm}
\usepackage{array}
\usepackage{ntheorem}
\usepackage{subfigure} 
\usepackage{cases}
\usepackage{stfloats}
\usepackage{url}
\usepackage{verbatim}
\usepackage{algpseudocode}  
\usepackage{color}
\newenvironment{sequation}{\begin{equation}\small}{\end{equation}}

\newenvironment{proof}{{\noindent\it Proof. }}{\hfill $\blacksquare$\par}
\usepackage{enumitem}[topsep=3pt,itemsep=3pt]

\allowdisplaybreaks[4]
\def\BibTeX{{\rm B\kern-.05em{\sc i\kern-.025em b}\kern-.08em
    T\kern-.1667em\lower.7ex\hbox{E}\kern-.125emX}}
\definecolor{color}{rgb}{0.0, 0, 0}
\definecolor{color1}{rgb}{0, 0, 0}
\definecolor{b}{rgb}{0.0, 0, 0}
\definecolor{r}{rgb}{0, 0, 0}
\ifCLASSOPTIONcompsoc
\usepackage[nocompress]{cite}
\else
\usepackage{cite}
\fi

\theoremseparator{.}
\newskip\theorempreskipamount
\newskip\theorempostskipamount
\newtheorem{theorem}{Theorem}
\newtheorem{definition}{Definition}
\newtheorem{lemma}{Lemma}
\theorembodyfont{}
\newtheorem{remark}{Remark}

\begin{document}

\title{Joint Task Offloading and Resource Allocation in Aerial-Terrestrial UAV Networks with Edge and Fog Computing for Post-Disaster Rescue}

\author{
    Geng~Sun,~\IEEEmembership{Member,~IEEE},
    Long~He,
    Zemin Sun,~\IEEEmembership{Member,~IEEE}, 
    Qingqing Wu,~\IEEEmembership{Senior Member,~IEEE}, 
    Shuang Liang,
    Jiahui Li,
    Dusit Niyato,~\IEEEmembership{Fellow,~IEEE}, and
    Victor C. M. Leung,~\IEEEmembership{Life Fellow,~IEEE}
		\thanks{This study is supported in part by the National Natural Science Foundation of China (62172186, 62002133, 61872158, 62272194), and in part by the Science and Technology Development Plan Project of Jilin Province (20230201087GX). (\textit{Corresponding authors: Zemin Sun and Shuang Liang.)}}
		\IEEEcompsocitemizethanks{\IEEEcompsocthanksitem Geng Sun, Long He, and Zemin Sun are with the College of Computer Science and Technology, Jilin University, Changchun 130012, China, and also with the Key Laboratory of Symbolic Computation and Knowledge Engineering of Ministry of Education, Jilin University, Changchun 130012, China.\protect\\
		E-mail: sungeng@jlu.edu.cn, helong0517@foxmail.com, sunzemin@jlu.edu.cn.
		\IEEEcompsocthanksitem Qingqing Wu is with the Department of Electronic Engineering, Shanghai Jiao
        Tong University, Shanghai, China.\protect\\
		E-mail: qingqingwu@sjtu.edu.cn.  
  	\IEEEcompsocthanksitem Shuang Liang is with the School of Information Science and Technology, Northeast Normal University, Changchun 130024, 
          China.\protect\\
		E-mail: liangshuang@nenu.edu.cn.
            \IEEEcompsocthanksitem Jiahui Li is with the College of Computer Science and Technology, Jilin University, Changchun 130012, China, and also with Pillar of Engineering Systems and Design, Singapore University of Technology and Design, Singapore 487372.\protect\\
            E-mail: lijiahui0803@foxmail.com.
  		\IEEEcompsocthanksitem Dusit Niyato is with the School of Computer Science and Engineering, Nanyang Technological University, Singapore 639798. \protect\\
            E-mail: dniyato@ntu.edu.sg. 
            \IEEEcompsocthanksitem Victor C. M. Leung is with the College of Computer Science and Software Engineering, Shenzhen University, Shenzhen 518060, China, and also with the Department of Electrical and Computer Engineering, University of British Columbia, Vancouver, BC V6T 1Z4, Canada. \protect\\
            E-mail: vleung@ieee.org.}
	\thanks{A small part of this paper appeared in IEEE MSN 2022~\cite{paper1}.}}

\markboth{}%
{Shell \MakeLowercase{\textit{et al.}}: Bare Demo of IEEEtran.cls for Computer Society Journals}

\IEEEtitleabstractindextext{%
\begin{abstract}		
\par Unmanned aerial vehicles (UAVs) play an increasingly important role in assisting fast-response post-disaster rescue due to their fast deployment, flexible mobility, and low cost. However, UAVs face the challenges of limited battery capacity and computing resources, which could shorten the expected flight endurance of UAVs and increase the rescue response delay during performing mission-critical tasks. To address this challenge, we first present a three-layer post-disaster rescue computing architecture by leveraging the aerial-terrestrial edge capabilities of mobile edge computing (MEC) and vehicle fog computing (VFC), which consists of a vehicle fog layer, a UAV client layer, and a UAV edge layer. Moreover, we formulate a joint task offloading and resource allocation optimization problem (JTRAOP) with the aim of maximizing the time-average system utility. Since the formulated JTRAOP is proved to be NP-hard,  we propose an MEC-VFC-aided task offloading and resource allocation (MVTORA) approach, which consists of a game theoretic algorithm for task offloading decision, a convex optimization-based algorithm for MEC resource allocation, and an evolutionary computation-based hybrid algorithm for VFC resource allocation. Simulation results validate that the proposed approach can achieve superior system performance compared to the other benchmark schemes, especially under heavy system workloads.
\end{abstract}
	
}

\maketitle

\IEEEdisplaynontitleabstractindextext

\IEEEpeerreviewmaketitle

%
%
\section{Introduction}
\label{sec:Introduction}
\par \IEEEPARstart{N}{atural} disasters, such as earthquakes, floods, and forest fires cause serious environmental damage, incalculable economic losses, and unpredictable loss of life~\cite{intro-A1}. Especially in cities, frequent man-made disasters, such as urban fires and traffic accidents, seriously affect people's quality of life and hinder social development. Although the occurrence of disasters can be reduced by deploying pre-disaster prevention facilities, it is critical to perform timely search-and-rescue missions within the golden window of post-disaster rescue. Therefore, effective post-disaster rescue requires immediate response to the disaster with rapid network deployment, real-time information collection, and low-latency data processing. However, the harsh conditions in the disaster-stricken area, such as the severely destroyed infrastructures of the terrestrial wireless networks and complex terrain, could interfere with rescue operations~\cite{intro-A2}. 

\par {\color{b} Due to the high maneuverability, flexible deployment, line-of-sight (LoS) communication, and low cost~\cite{GuoLLTK22,guo2023intelligent}, it is practical to dispatch unmanned aerial vehicles (UAVs) to the affected area to assist in rescue missions such as disaster area monitoring, data collection, and aerial search and rescue~\cite{HayatYM16}. These rescue missions often require UAVs to perform compute-intensive computing tasks such as video processing, data analysis, and feature extraction, with strict latency requirements~\cite{introduction-2,introduction-3}. Such cases commonly existed in UAV-assisted post-disaster rescue scenarios. For instance, in~\cite{xu2023reward}, UAVs are deployed to monitor points of interest (PoIs) within a disaster-stricken region, where a PoI may denote an office building or a school building, in which people are trapped. These UAVs are required to process the collected images and data in real time and send their collected information back to the ground rescue center. In~\cite{dong2021uav}, UAVs are employed to perform aerial search and rescue (SAR) missions. These UAVs equipped with thermal cameras can capture thermal image data of victims and are required to promptly process these data sets through neural networks to effectively locate and assist victims.}


\par However, the constrained onboard battery capacity and computing capability of UAVs could restrict the endurance and efficiency in performing the intensive computing tasks of rescue missions. Fortunately, {\color{color} task offloading has emerged as a promising solution} to extend the capabilities of devices with limited computing resources and energy. A number of studies employ mobile edge computing (MEC)~\cite{Du2022,lyu2020lead}, cloud computing, or vehicle fog computing (VFC) to reduce the latency and energy consumption of task processing~\cite{intro-mec}. However, these studies may not be directly applicable to multi-UAV-assisted post-disaster rescue scenarios for the following reasons. First, the fast deployment of traditional terrestrial MEC servers could be an arduous task in disaster areas with complex terrains since the terrestrial MEC servers rely heavily on communication infrastructures. Second, due to the remote location of cloud servers, cloud computing suffers from large transmission latency, which is difficult to fulfill the low-latency demands of the computing tasks for rescue missions. Third, VFC depends on the idle resources and positions of vehicles, resulting in the unstable availability of the VFC resources.

\par {\color{color} Inspired by the current cloud-edge-device structure of 5G, we innovatively propose a three-layer post-disaster rescue computing architecture, which can be aligned and compatible with the 5G cloud-edge-device architecture.} 
Specifically, the three-layer computing architecture is comprised of a UAV edge layer, a UAV client layer, and a vehicle fog layer, which leverages MEC and VFC to exploit the aerial and terrestrial resources of the UAV network. On one hand, the aerial MEC capability provided by the edge UAV makes up for the unstable availability of the VFC resources provided by the vehicle fog nodes. On the other hand, the terrestrial VFC capability can effectively alleviate the overload caused by the limited computing resources of the edge UAV.

\par However, several fundamental challenges should be overcome to fully develop the benefits of integrating UAVs, MEC, and VFC techniques for post-disaster rescue. \textbf{i)} The offloading decision of each UAV depends not only on its own offloading demand but also on the offloading decisions of the other UAVs, which makes the offloading decisions among UAVs coupling and complex. \textbf{ii)} Various tasks of UAVs generally arrive dynamically and have stringent requirements for the offloading service. {\color{color1}However, the limited computing resource of an MEC server and the stringent demands of the UAVs could lead to competition for resources inside the MEC server, especially during peak times. Thus, under the resource constraint, it is challenging for the MEC server to determine an efficient resource allocation strategy to meet the heterogeneous and stringent demands of various tasks.} \textbf{iii)} The computing resource allocation strategy of the MEC server and the task offloading decisions of UAVs have mutual effects on each other, leading to the complexity of the decision-making. \textbf{iv)} {\color{color} Due to the mobility of rescue vehicles and UAVs and the insufficient time-varying idle computing resources that vehicles provide, it is difficult to design an effective VFC approach to fully utilize the resources of rescue vehicles.}

\par To overcome the aforementioned challenges, we propose an approach for joint optimization of MEC-VFC-aided task offloading and resource allocation to maximize the system performance. The main contributions are summarized as follows:
 
\begin{itemize}

\item \textit{System Architecture.} We employ a three-layer post-disaster rescue computing architecture in the MEC-VFC-aided aerial-terrestrial UAV network to coordinate UAVs and ground rescue vehicles to cooperatively process the computing tasks. Specifically, the proposed architecture consists of the following entities: a UAV edge layer where a large UAV is deployed to provide aerial edge capability, a UAV client layer where several small UAVs are deployed to perform the computing tasks, and a vehicle fog layer where the ground rescue vehicles with under-utilized resources are leveraged to provide terrestrial edge capability to alleviate the possible computational overload of the UAV edge layer.

\item \textit{Problem Formulation.} We formulate a novel joint task offloading and resource allocation optimization problem (JTRAOP), with the aim of maximizing the time-average system utility. Specifically, the system utility function is theoretically constructed by synthesizing the completion delay of the tasks and the energy consumption of UAVs.
	
\item \textit{Algorithm Design.} Due to the NP-hardness of JTRAOP, we propose an MEC-VFC-aided task offloading and resource allocation (MVTORA) approach that consists of two components, i.e., task offloading and computing resource allocation to solve the problem separately. For task offloading, {\color{color} we propose a game theoretic algorithm to determine the task offloading decisions}. For computing resource allocation, an convex optimization-based algorithm and an evolutionary computation-based hybrid algorithm are proposed to determine the aerial MEC resource allocation and terrestrial VFC resource allocation, respectively. The proposed MVTORA approach is theoretically proved to be stable and has polynomial computation complexity.

\item \textit{Validation.} Simulation results demonstrate that the proposed MVTORA is able to achieve superior performance in terms of time-average system utility, average task completion delay, and total energy consumption compared to several benchmark schemes. In addition, through the simulation results, we found that the proposed algorithm not only has better scalability but also can achieve significant system performance improvement when the workload is heavy.
	
\end{itemize}

\par The remainder of the article is organized as follows. Section \ref{sec:Related Works} summarizes the related work. In Section \ref{sec:System Model and Problem Formulation},  the proposed system model and problem formulation are presented. Section \ref{sec:MVTORA} proposes the MVTORA approach. In Section \ref{sec:Simulation Results and Analysis}, simulation results are displayed and analyzed. {\color{b} In Section \ref{sec:Discussion}, we present related discussions.} Finally, Section \ref{sec:Conclusion} concludes the overall paper. 
%
%
\section{Related Works}
\label{sec:Related Works}

\par In this section, we review the research work related to UAV-assisted disaster rescue, task offloading and resource allocation.
%
%
\subsection{UAV-Assisted Disaster Rescue}
\label{subsec:Applications of UAVs}
\par Due to the swift and affordable deployment, high mobility, and better LoS communication links, UAVs show great potential in aiding post-disaster rescue. Extensive efforts have been exerted to exploit the advantages of UAVs to provide flexible wireless connectivity~\cite{review-A1,review-A3,Kang2023Personalized} and (or) computing capability \cite{review-A6,guo2023multi,guo2022achieve}. For example, Zhang et al.~\cite{review-A1} focused on UAV-enabled radio communications where two UAVs were deployed to relay satellite data to ground vehicles. Mozaffari et al.~\cite{review-A3} explored the scenario where UAVs were deployed as aerial base stations to collect information on ground Internet of Things (IoT) devices. {\color{b}Kang et al.~\cite{Kang2023Personalized} focused on UAV image-sensing-driven task oriented semantic communications scenarios, where the effects of dynamic wireless fading channel on semantic transmission mathematically were studied.} Tong et al.~\cite{review-A6} proposed a UAV-enabled multi-hop cooperative fog computing system, where UAVs assist to provide wireless communication and computing services simultaneously to the user equipment. {\color{b}Guo et al.~\cite{guo2023multi} considered multiple UAV-enabled aerial computing to break the computing resources and coverage limitations of a single UAV. In~\cite{guo2022achieve}, the authors studied a dynamic UAV edge computing IoT network framework, aiming to stably provide fast communication and computing services over a long period of time.} {\color{b} The above work has undoubtedly
explored the feasibility and effectiveness of UAV providing communication and computing capabilities. However, these studies are mainly devoted to exploiting the communication and computation capabilities of UAVs as service entities, while their limited resources and battery capacity as client entities are not extensively considered.}

\par Different from the traditional wireless communication scenarios, in the post-disaster rescue scenarios, the UAVs are deployed as client entities to assist in various rescue missions that generate various delay-sensitive and computation-intensive computational tasks. However, limited by the computation processing capability and onboard battery power, it is difficult for UAVs to satisfy the stringent requirements of rescue missions. To deal with the limited resources, several studies target on exploring the under-utilized fog resources. For example, in~\cite{review-A7}, a swarm of UAVs with idle resources was empowered by the capability of edge computing to act as a UAV fog. The task generated by a UAV is divided into multiple sub-tasks where one part is processed locally and the other parts are offloaded to the nearby fog UAVs. Similarly, in~\cite{review-A8}, by exploring the available resources of unmanned ground vehicles (UGVs), the computing tasks of UAVs can be offloaded to the VFC nodes composed of UGVs. However, the UAV fog is also limited by the restricted computation resources and battery storage, and the mobility of the VFC node constrains the efficiency of conducting urgent rescue tasks.
%
%
\subsection{Task Offloading}
\label{subsec:Computation offloading in MEC and VFC}
\par Task offloading has been widely studied as an effective way to alleviate the limited computing resources and energy supply of devices. In the past ten years, numerous research studies have focused on offloading computation-intensive tasks to remote cloud servers, which provide sufficient computing resources~\cite{Hoang2013}. For instance, Abd et al.~\cite{review-B1-1} developed a cloud computing model that supports on-demand internet access and ubiquitous access to different resource bundles for configurable computing. {\color{color1} Deng et al.~\cite{review-B1-2} investigated the mobile computation offloading where the complex requirements of multiple mobile services in workflows are fulfilled by offloading the services to the cloud.} However, the location of cloud servers far from end users often leads to unacceptable delay and heavy backhaul utilization. 
\par MEC empowers the provision of cloud-computing resources for end-users in close proximity to the network edge, which greatly improves the quality of experience (QoS) for end-users. Chen et al.~\cite{review-B2-1} considered the multi-user binary computing offloading in an MEC-enabled network with multi-channel radio interference, where a game-theoretic scheme is developed to achieve efficient task offloading in a distributed manner. Hekmati et al.~\cite{review-B2-2} focused on an MEC scenario with hard-deadline constrained tasks, and a framework of joint remote and local concurrent execution is proposed to ensure the hard task completion deadlines of tasks. Goudarzi et al.~\cite{review-B2-3} explored the application placement in computing environments with multiple centralized cloud servers and distributed fog/edge servers. Specifically, a weighted cost model is proposed to improve the execution time and energy consumption of devices. 
\par As an extension of MEC, the concept of fog computing has been proposed to include devices with idle computing resources, such as smartphones, set-top boxes, etc~\cite{review-B3-1}. With the explosive growth of vehicles and the rapid development of vehicle network techniques~\cite{GuoZLZ22}, VFC has become a research hotspot in both academia and industry. For example, Feng et al. \cite{review-B3-2} proposed a framework to extend the computing capabilities of vehicles in a distributed manner by utilizing the idle computing resources of other vehicles. Zhu et al.~\cite{review-B3-3} are devoted to a solution for latency and quality-balanced task allocation in the VFC-assisted network. {\color{color} Specifically, a joint optimization problem is formulated for the task assignment between mobile and stationary fog nodes, which is solved by mixed integer linear programming.}
%
%
\subsection{Resource Allocation}
\label{subsec:Resource Allocation}
\par Considering the limited communication and computing resources of the edge servers for MEC technology, the resource allocation for task offloading has received widespread attention. There are some works that only optimize radio resource allocation or computing resource allocation. For example, You et al.~\cite{You2017} considered wireless resource allocation based on orthogonal frequency-division multiple access (OFDMA) and time-division multiple access (TDMA) respectively to minimize energy consumption under delay constraints. {\color{color} Yang et al.~\cite{Yang2019} proposed an iterative algorithm to optimize resource allocation based on non-orthogonal multiple access (NOMA), where the total energy consumption and completion delay of tasks are significantly improved. Considering the limited computing and storage capabilities of a single-edge server, Ning et al.~\cite{Ning2021} studied multi-server cooperation to maximize the system utility. To improve the efficiency of resource utilization in an MEC system. {\color{b} Chen et al.~\cite{Chen2022} proposed a two-level adaptive resource allocation framework to support vehicular safety message transmissions.}}

\par There are some works that focus on joint optimization of communication and computing resource allocation~\cite{Liu2021,JosiloD19,ZhangGHCZT20,guo2021inter}. In~\cite{Liu2021}, the authors studied NOMA-enabled massive IoT networks to jointly optimize radio and computation resource allocation with the aim of maximizing the energy efficiency for task offloading. In~\cite{JosiloD19}, a Stackelberg game is proposed to model the interaction between the allocation of wireless and computing resource allocation from the edge cloud operator and the task offloading of users. In~\cite{ZhangGHCZT20}, the authors jointly optimized the task offloading decision, wireless resource allocation, and computation resource scheduling in an MEC-enabled dense cloud radio access network (C-RAN), and the Lyapunov optimization theory is developed for online decision-making. {\color{b}In~\cite{guo2021inter}, the authors considered the huge demands for network access and computing resources in ultra-dense edge computing and proposed an inter-server collaborative federated learning method to reduce both training and communication time. }
\par However, due to the unpredictable computing or communication demand, unstructured networks, and stringent QoS requirements of rescue missions, few of the above works study the problem of task offloading and resource allocation in multi-UAV-assisted post-disaster rescue scenarios. {\color{b}Distinguished from the previous studies above, this work proposes a three-layer post-disaster rescue computing architecture, which combines with UAV-enabled MEC and vehicle-enabled VFC. Moreover, we study the problem of task offloading and resource allocation based on the three-tier architecture to jointly exploit heterogeneous computing resources of MEC and VFC.}

 
%
%
\section{System Architecture and Problem Formulation}
\label{sec:System Model and Problem Formulation}
\par In this section, we first propose the hierarchical computing architecture for multi-UAV-assisted post-disaster rescue, and illustrate the related system models. Then, we formulate the joint optimization problem studied in this work.
%
%
\subsection{System Architecture}
\label{subsec:System Model}

\subsubsection{System Overview}
\label{subsec:System Overview}

\par As shown in Fig. \ref{fig_gameModel}, in the spatial domain, we consider a three-layer MEC-VFC-assisted post-disaster rescue computing architecture consisting of a vehicle fog layer, a UAV client layer, and a UAV edge layer. {\color{b} The architecture consists of three types of entities, i.e., rescue vehicles, small rotary-wing UAVs, and a large rotary-wing UAV. Among these entities, the small rotary-wing UAVs are the served entities, while the rescue vehicles and large rotary-wing UAV are the serving entities, which are detailed as follows.}


\par \textit{At the vehicle fog layer}, rescue vehicles are deployed in the disaster area for post-disaster recovery and reconstruction. Furthermore, these rescue vehicles can act as fog nodes to share idle computing resources with UAVs at the UAV client layer to alleviate the possible overload of the edge UAV at the UAV edge layer. 

\par \textit{At the UAV client layer}, the small rotary-wing UAVs are deployed as client UAVs (C-UAVs) to assist in aerial rescue missions according to the pre-set trajectories, which generate computation-intensive and latency-sensitive computing tasks. Furthermore, each C-UAV is responsible for a given area without overlapping the service area of the adjacent C-UAVs~\cite{review-A8}. Besides, each C-UAV flies at a fixed altitude $H$ to avoid the energy consumption caused by frequent aircraft ascending and descending due to obstacles~\cite{systemmodel-H}. Besides, due to the limited computing resources and battery capacity, each C-UAV independently decides whether to process the tasks locally (referred as local computing), offload the tasks to the edge UAV (referred as MEC-assisted offloading), or offload the tasks to vehicle fog nodes (referred as VFC-assisted offloading).

\par \textit{At the UAV edge layer}, a large rotary-wing UAV equipped with the MEC capability is deployed as an edge UAV (E-UAV) at the center of the disaster area with the following functionalities: 1) providing wireless communication coverage for the C-UAVs and rescue vehicles; 2) providing computation resources for the C-UAVs; 3) informing the ground control center of the on-site information on the disaster area; and 4) acting as a regional controller to make decisions by running algorithm via the collected channel state information (CSI) and state information of rescue vehicles and C-UAVs. 
\begin{figure*}[!hbt] 
	\centering
	\includegraphics[width =6.5in]{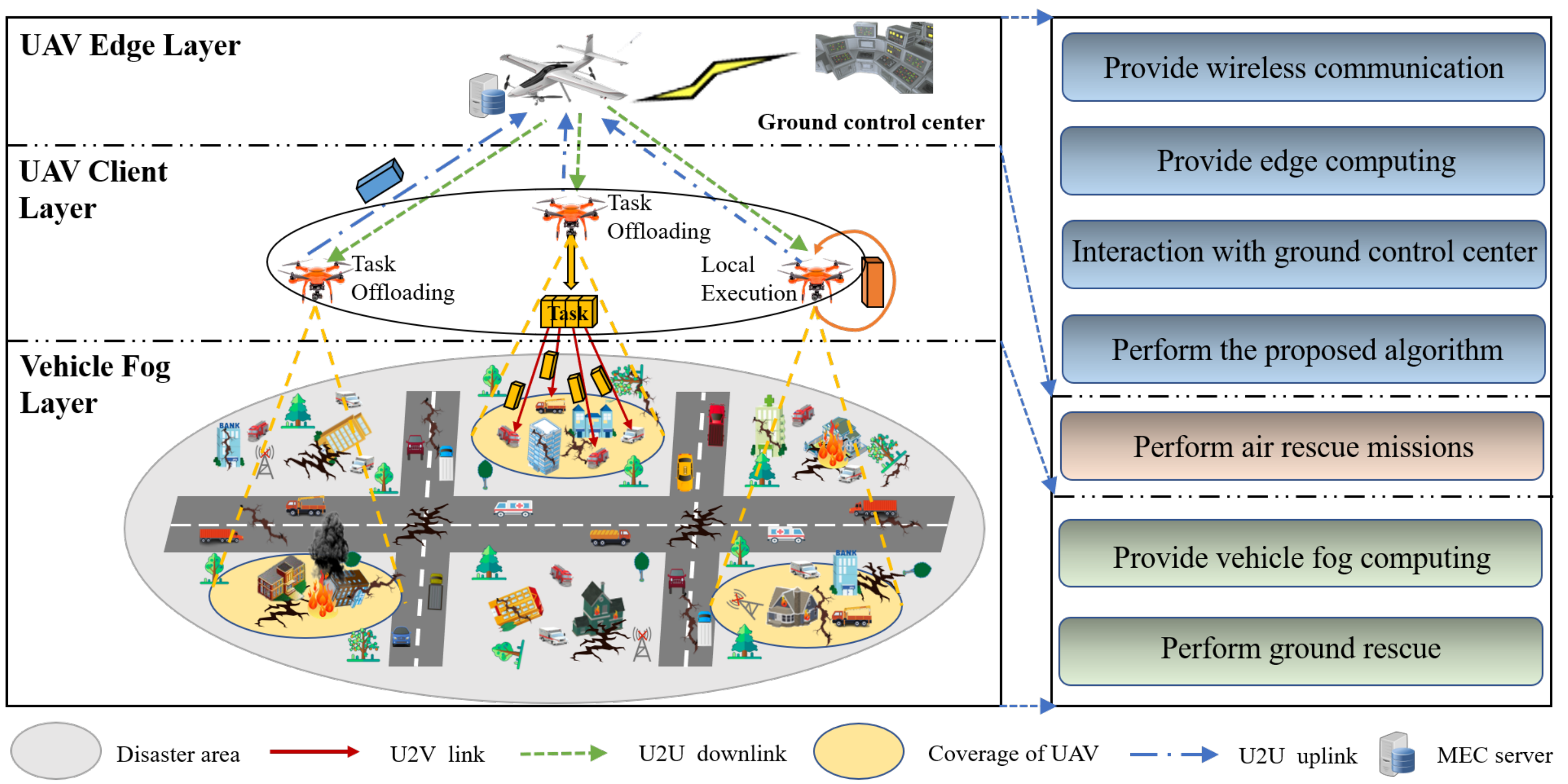}
	\caption{The MEC-VFC-aided aerial-terrestrial UAV network consists of a large rotary-wing UAV, a group of small rotary-wing UAVs, and multiple rescue vehicles. Each small rotary-wing UAV either computes its task locally or offloads the task to the larger UAV or divides the task into multiple sub-tasks and offloads them to rescue vehicles.}
	\label{fig_gameModel}
\end{figure*}
\par In the temporal domain, the system timeline is discretized into equal $T$ time slots~\cite{wang2023service}, i.e., $t\in \mathcal{T}=\{0,\ldots,T-1\}$, wherein each slot duration is $\Delta t$. In each time slot, the CSI and the state information of both rescue vehicles and C-UAVs are captured and updated by the E-UAV, where the corresponding decisions are determined by running our algorithm.

\begin{remark}
\par Note that our current work does not address the optimization of UAV trajectories. The main reason is that the trajectory planning of UAVs for post-disaster rescue relies on specific rescue missions and the terrain of post-disaster scenarios, which is independent of the task offloading and resource allocation problems in this work.
\end{remark}
\subsubsection{Basic Models}

\par The basic models of entities in the system are shown as follows. 

\textbf{\textit{Vehicle Fog Model.}} The set of rescue vehicles is denoted as $\mathcal{M} = \{1, \ldots, M\}$. Each vehicle $m\in\mathcal{M}$ is characterized by $\mathbf{St}_m^{\text{veh}}(t)=\left(\mathbf{P}_m(t),v_m(t),\theta_m(t),f_m^{\text{veh}}(t) \right)$, where $\mathbf{P}_m(t)=\left[x_m(t),y_m(t),0\right]$, $v_m(t)$, $\theta_m(t)$ and $f_m^{\text{veh}}(t)$ denote the position, velocity, direction, and idle computing resources of vehicle $m$ at time $t$, respectively. {\color{b} We consider that the vehicles are distributed in the disaster area following a Poisson point process (PPP) with density $\rho_v$. Moreover, by using the Gauss-Markov model~\cite{LiangH99},
the mobility of the vehicles is modeled as a temporal-dependent process, which is given as follows:}
\begin{equation}
\begin{split}
    &v_m(t+1)=\alpha v_m(t)+(1-\alpha)\overline{v}+\sqrt{1-\alpha^2}\omega_t^v,\\
    &V_{\text{veh}}^{\text{min}}\leq v_m(t)\leq V_{\text{veh}}^{\text{max}},\ \forall m\in \mathcal{M},\ \forall t\in \mathcal{T},
    \label{eq.vehcle_velocity}
\end{split}
\end{equation}
\noindent where $v_m(t+1)$ is the velocity of vehicle $m$ at time $t+1$ and {\color{color} $\omega_t^v$ is the uncorrelated random Gaussian process with mean $0$ and the asymptotic variance of velocity $\sigma_v^2$. Furthermore, $\alpha$ and $\overline{v}$ denote the memory degree and asymptotic mean of velocity, respectively.} {\color{color1} Similarly, direction $\theta_m$ can be given as:}
\begin{equation}
\begin{split}
    &\theta_m(t+1)=\alpha \theta_m(t)+(1-\alpha)\overline{\theta}+\sqrt{1-\alpha^2}\omega_t^d,\\
    &\Theta_{\text{veh}}^{\text{min}}\leq \theta_m(t)\leq \Theta_{\text{veh}}^{\text{max}},\ \forall m\in \mathcal{M},\ \forall t\in \mathcal{T},
    \label{eq.vehcle_direction}
\end{split}
\end{equation}
\noindent where $\theta_m(t+1)$ is the direction at time $t+1$ and $\omega_t^d$ is the uncorrelated random Gaussian process with mean $0$ and the asymptotic variance of direction $\sigma_d^2$. Furthermore, $\overline{\theta}$ represents the asymptotic mean of direction. Therefore, the mobility of vehicle $m$ can be updated as:
\begin{equation}
\begin{split}
    &x_m(t+1)=x_m(t)+v_m(t)\cdot \cos(\theta_m(t))\cdot \Delta t,\\
    &y_m(t+1)=y_m(t)+v_m(t)\cdot \sin(\theta_m(t))\cdot \Delta t.
    \label{eq.vehcle_mobility}
\end{split}
\end{equation}

\begin{remark}
\par {\color{b} Among the existing models, modeling vehicle mobility using PPP and Gaussian Markov models is more realistic in post-disaster scenarios. \textbf{\textit{First}}, due to the random nature of rescue point distribution and the constraint on vehicle movement caused by the random road damage resulting from disasters, the distribution of vehicles exhibits a certain degree of randomness. PPP is a commonly used model to characterize the distribution characteristics of rescue vehicles in disaster scenarios~\cite{review-A1,HouraniSJ16}. \textbf{\textit{Second}}, the rescue vehicles deployed to perform ground rescue missions usually travel toward a destination, and therefore the rescue vehicle’s location and velocity in the future are likely to be correlated with its current location and velocity. The Gaussian Markov mobility model is proposed to capture the essence of temporal-dependent process~\cite{LiangH99}. Thus, we adopt PPP and Gauss-Markov models to describe vehicle mobility, which provides a useful balance between realism and tractability.}
\end{remark}

\par \textbf{\textit{C-UAV Model.}} The set of C-UAVs is denoted as $\mathcal{N} = \{1, \ldots, N\}$. Each C-UAV $n\in \mathcal{N}$ is characterized by $\mathbf{St}_n^{\text{uav}}(t)=\left(\mathbf{P}_n(t),v_n(t),\theta_n(t),g_n(t),\mathbf{\Phi}_n(t),f_n^{\text{uav}}\right)$, where $\mathbf{P}_n(t)=\left[x_n(t),y_n(t), H\right]$, $v_n(t)$, and $\theta_n(t)$ respectively denote the position, velocity, and direction of C-UAV $n$ at time $t$, which are known according to the pre-set trajectory. Moreover, we assume that each C-UAV can generate multiple tasks within the system timeline and at most one task in each time slot~\cite{task-arrival}. {\color{b} Specifically, the computing tasks generated by C-UAVs are modeled as an independent and identically distributed Bernoulli process~\cite{ChenWCLZBLJ22,duan2022moto}. For each C-UAV $n$, a computing task is generated with probability $\rho_n\in [0,1]$ at the beginning of each slot. Moreover, $g_n(t)\in \{0,1\}$ is a binary variable to indicate whether C-UAV $n$ generates a task at time $t$, where $g_n(t)=1$ means that C-UAV $n$ generates a task. Then, $\mathbb{P}(g_n(t)=1)=1-\mathbb{P}(g_n(t)=0)=\rho_n$, where $\mathbb{P}(.)$ denotes the probability of an event occurring.} $\mathbf{\Phi}_n(t)=\{D_n(t),\eta_n(t),T^{\text{max}}_n(t)\}$ represents the computing task generated by C-UAV $n$ at time $t$, wherein $D_n(t)$ presents the data size of the input task (in bits), $\eta_n(t)$ is the computation intensity of the task (in cycles/bit), and $T_n^{\text{max}}(t)$ denotes the maximum acceptable delay of the task. The local computation capability of C-UAV $n$ is denoted as $f_n^{\text{uav}}$. In addition, we define a binary variable $a_n^i(t)\in \{0,1\}$ ($i\in \mathcal{I}=\{\text{loc},\text{mec},\text{veh}\}$) to represent the offloading decision of C-UAV $n$ at time $t$, wherein $a_n^{\text{loc}}(t)=1$ implies the task is executed locally on C-UAV $n$, $a_n^{\text{mec}}(t)=1$ implies the task is executed on the E-UAV, $a_n^{\text{veh}}(t)=1$ implies the task is executed on vehicle fog nodes, and $a_n^{\text{loc}}(t)+a_n^{\text{mec}}(t)+a_n^{\text{veh}}(t) = 1$, respectively.

\par \textbf{\textit{E-UAV Model.}} The E-UAV $u$ hovering over the disaster area is characterized by $\mathbf{St}^{u}=(\mathbf{P}_{u},F_{u}^{\text{max}})$, wherein $\mathbf{P}_{u}=[x_u,y_u,H_u]$ and $F_{u}^{\text{max}}$ denote the position and the maximal computing resources of the E-UAV, respectively. 

%
%
\subsection{Communication Model}
\label{subsec:Communication Model}

\par The C-UAVs can decide to offload the tasks to vehicle fog nodes or E-UAV via UAV-to-vehicle (U2V) links and UAV-to-UAV (U2U) links, respectively, and the widely used OFDMA is employed in the communication models. Specifically, for each C-UAV $n$, there are $K_n$ orthogonal wireless sub-channels~\cite{sub-channel}. {\color{b} Furthermore, we assume that each C-UAV is equipped with a directional antenna of adjustable beamwidth and the azimuth and elevation half-power beamwidths of the antenna are equal, which is presented by $2\Psi\in (0,\pi)$~\cite{antenna}}. Therefore, the antenna gain of C-UAV $n$ in the direction with azimuth $\psi^a$ and elevation $\psi^e$ can be obtained as~\cite{commodel-gain}:
\begin{equation}
    G(\psi^a,\psi^e) = 
    \begin{cases}
        \frac{G_0}{\Psi^2}, &-\Psi<\psi^a<\Psi,-\Psi<\psi^e<\Psi\\
        g\approx0, &\text{otherwise},
    \end{cases}
    \label{eq.gain}
\end{equation}
\noindent where $g$ denotes the channel gain outside the beamwidth of the antenna. {\color{b} In practice, $0<g\ll G_0/\Psi^2$~\cite{LyuZ19}, which means that the communication links outside the beamwidth of the antenna is difficult to meet the communication requirements. Therefore, for simplicity, we set $g=0$ to indicate that we do not consider communication links outside the beamwidth of the antenna.}

\par \textbf{\textit{U2V Communication.}} This work employs a probabilistic LoS channel model for the communication between C-UAVs and vehicles \cite{communicationmodel2}. The channel coefficient between C-UAV $n$ and vehicle $m$ at time $t$ can be presented as follows~\cite{UAV-review}:
\begin{equation}
	h_{n,m}(t)=\sqrt{\beta_{n,m}(t)}\Tilde h_{n,m}(t),
	\label{eq-h}
\end{equation}
\noindent where $\tilde h_{n,m}(t)$ represents the coefficient of small-scale fading that is generally a complex random variable with $E[|\tilde h_{n,m}(t)|^2]=1$, and $\beta_{n,m}(t)$ denotes the coefficient of large-scale fading that includes the distance-dependent path loss and shadowing. For the U2V links, the large-scale fading is generally modeled as a random variable that depends on the occurrence probabilities of LoS and non-line-of-sight (NLoS) links, which is given as~\cite{communicationmodel1}:
\begin{equation}
	\beta_{n,m}(t)=
	\begin{cases}
		\beta_0d_{n,m}^{-\mu}(t), &\text{LoS link},\\
		\kappa\beta_0d_{n,m}^{-\mu}(t), &\text{NLoS link},
	\end{cases}
	\label{eq.LoS-and-nLoS}
\end{equation}
\noindent where $\beta_0$ denotes the constant path loss coefficient at the reference distance of 1 m under the LoS condition, $d_{n,m}(t)=\|\mathbf{P}_{n}(t)-\mathbf{P}_{m}(t)\|$ denotes the straightline distance between C-UAV $n$ and vehicle $m$ at time $t$, $\mu$ is the path loss exponent, and $\kappa<1$ is the additional attenuation factor due to the NLoS propagation. 

\par The LoS probability $P_{n,m}^{\text{LoS}}(t)$ between C-UAV $n$ and vehicle $m$ is generally modeled as a logistic function of the elevation angle $\theta_{n,m}(t)$, which is given as ~\cite{communicationmodel5}:
\begin{equation}
	P_{n,m}^{\text{LoS}}(t)=\frac{1}{1+a \exp(-b(\theta_{n,m}(t) - a))},\label{eq.P_LoS}
\end{equation}
\noindent where $a$ and $b$ are constants that depend on the propagation environment, and $\theta_{n,m}(t) =\frac{180}{\pi
} \arcsin{\frac{H}{d_{n,m}(t)}}$ denote the elevation angle in degree. Therefore, the expected channel power gain can be given as:
\begin{sequation}
\label{eq}
    E[|h_{n,m}(t)|^2]= P_{n,m}^{\text{LoS}}(t)\beta_0d_{n,m}^{-\mu}(t)+(1-P_{n,m}^{\text{LoS}}(t))\kappa \beta_0d_{n,m}^{-\mu}(t).
\end{sequation}

\noindent Furthermore, we assume that the change of the LoS probability between C-UAV $n$ and vehicle $m$ within a time slot can be negligible because the time slot duration is set small enough~\cite{parameter_2}. Then, the average communication rate between C-UAV $n$ and vehicle $m$ at time $t$ is described as follows:
\begin{equation}
    R_{n,m}(t)=B\log_2\left(1+\frac{P_{n,m}E[|h_{n,m}(t)|^2]G_0}{\Psi^2\sigma^2B}\right), \label{eq.transratio1}
\end{equation}
\noindent where $B$ denotes the bandwidth of the sub-channel, {\color{b} $P_{n,m}$ represents the transmission power between C-UAV $n$ and vehicle $m$}, and $\sigma^2$ is the noise power spectral density.

\par \textbf{\textit{U2U Communication.}} The U2U communication is characterized by the free-space path-loss model since it is dominated by LoS links. The average communication rate between C-UAV $n$ and E-UAV $u$ is given as follows:
{\color{b}
\begin{equation}
    R_{n,u}(t)=K_nB\log_2\left(1+\frac{P_{n,u}\tilde{\beta}_0G_0 d_{n,u}^{-2}(t)}{\Psi^2\sigma^2K_nB}\right), \label{eq.transratio}\\
\end{equation}}
\noindent where $P_{n,u}$ is the transmission power from C-UAV $n$ to E-UAV $u$, {\color{b}$\tilde{\beta}_0$ is the channel power gain at the reference distance}, and $d_{n,u}(t)=\|\mathbf{P}_n(t)-\mathbf{P}_u(t)\|$ is the distance between C-UAV $n$ and E-UAV $u$ at time $t$.

\subsection{Service Delay and Energy Consumption}
\label{subsec:Computation Model}

\par The service delay and energy consumption to complete task $\mathbf{\Phi}_n(t)$ depend on the offloading strategy $a_n^i(t)$ of C-UAV $n$.

\par \textbf{\textit{Local Computing.}} When task $\mathbf{\Phi}_n(t)$ is executed on C-UAV $n$ locally (i.e., $a_n^{\text{loc}}(t)=1$), the local service delay at time $t$ can be calculated as:
\begin{equation}
	T_n^{\text{loc}}(t)=\frac{\eta_n(t)D_n(t)}{f_n^{\text{uav}}}. \label{eq.loc-computation}
\end{equation}
\par Correspondingly, the energy consumption of C-UAV $n$ to execute task $\mathbf{\Phi}_n(t)$ locally at time $t$ can be calculated as~\cite{UAV-H}:
\begin{equation}
	E_n^{\text{loc}}(t)=k(f_n^{\text{uav}})^{3}T_n^{\text{loc}}(t),\label{eq.loc-energy-consumption}
\end{equation}
\noindent where $k$ is the effective switched capacitance cofficient for each C-UAV that depends on the hardware architecture~\cite{CPU-Power}.

\textbf{\textit{MEC-Assisted Offloading.}} When task $\mathbf{\Phi}_n(t)$ is offloaded to the E-UAV for execution (i.e., $a_n^{\text{mec}}(t)=1$), the service delay of task at time $t$ includes the transmission delay and the E-UAV execution delay, which can be given as:
\begin{equation}
\begin{split}
    \label{eq.mec-computation}
    T_n^{\text{mec}}(t) = \frac{D_n(t)}{R_{n,u}(t)}+\frac{\eta_n(t) D_n(t)}{F_n(t)},
\end{split}
\end{equation}
\noindent where $F_n(t)$ denotes the computing resource allocated by the E-UAV to task $\mathbf{\Phi}_n(t)$.

\par The energy consumption of C-UAV $n$ to offload task $\mathbf{\Phi}_n(t)$ to the E-UAV is mainly induced by the task transmission, which can be given as~\cite{10106022}: 
\begin{equation}
	E_n^{\text{mec}}(t)=\frac{P_{n,u} D_n(t)}{R_{n,u}(t)}.\label{eq.mec-energy-consumption}
\end{equation}
\par \textbf{\textit{VFC-Assisted Offloading.}} When task $\mathbf{\Phi}_n(t)$ is offloaded to vehicle fog nodes for execution (i.e., $a_n^{\text{veh}}(t)=1$), we consider that task $\mathbf{\Phi}_n(t)$ $(n\in \mathcal{N})$ can be divided into multiple independent sub-tasks owing to the insufficient computing resources of vehicle fog nodes \cite{systemmodel-3}, and the time for task division is short enough to be negligible~\cite{systemmodel-taskdivide}. Furthermore, these sub-tasks can be offloaded by C-UAV $n$ to the set of rescue vehicles within its communication range (i.e., $d_{n,m}(t)\leq H\tan\Psi$) for parallel processing. Due to the limited number of sub-channels, C-UAV $n$ can offload sub-tasks to $K_n$ vehicles at most simultaneously. Therefore, we define $\mathbf{S}_n^{\prime}(t)$ as the set of vehicles selected by C-UAV $n$ to perform sub-tasks and $\mathbf{\lambda}_n^t=\{\lambda_{n,j}^t\}_{j\in \mathbf{S}_n^{\prime}(t)}$ as the task division set of C-UAV $n$ at time $t$, wherein $\lambda_{n,j}^t$ is the proportion of sub-task offloaded to vehicle $j$ in the total task, $|\mathbf{S}_n^{\prime}(t)|\leq K_n$, and $\lambda_{n,j}^t\in [0,1]$. Therefore, the service delay of task at time $t$, including the transmission delay and the vehicle execution delay, can be calculated as:
\begin{equation}
	T_n^{\text{veh}}(t)=\max\limits_{j\in \mathbf{S}_n^{\prime}(t)}\left(\frac{\lambda_{n,j}^t D_n(t)}{R_{n,j}(t)}+\frac{\lambda_{n,j}^t \eta_n(t)D_n(t)}{f_j^{\text{veh}}(t)}\right),\label{eq.veh-delay}
\end{equation}
\noindent where $\mathbf{S}_n(t)\ (\mathbf{S}_n^{\prime}(t)\subset \mathbf{S}_n(t))$ is the set of vehicles within C-UAV $n$ communication range at time $t$, $\sum_{j\in \mathbf{S}_n^{\prime}(t)}\lambda_{n,j}^t=1$ and $f_j^{\text{veh}}(t)$ denotes the idle computing resources owned by vehicle $j$ at time $t$.

\par The energy consumption of C-UAV $n$ to offload task $\mathbf{\Phi}_n(t)$ to vehicles is mainly induced by the task transmission, which can be given as: 
\begin{equation}
\begin{split}
    \label{eq.veh-energy-consumption}
    E_n^{\text{veh}}(t)&=\sum_{j\in \mathbf{S}_n^{\prime}(t)}P_{n,j}T_{n,j}^{\text{veh}}(t)=\sum_{j\in \mathbf{S}_n^{\prime}(t)}\frac{P_{n,j}\lambda_{n,j}^t D_n(t)}{R_{n,j}(t)}.
\end{split}
\end{equation}
\begin{remark}
\par When calculating the energy consumption of C-UAVs, the propulsion energy consumption of C-UAVs is omitted. This is because the C-UAVs fly according to the pre-set trajectories, leading to constant propulsion energy consumption, which would have no effect on the results of decision-making for C-UAVs.
\end{remark}
%
%
\subsection{Utility Function}
\label{subsec:Utility Function}

\par In this sub-section, the utility function is formulated to quantify the satisfaction level of C-UAVs in performing tasks, which can be formulated by considering the following metrics.

\textbf{\textit{Revenue of Task Processing.}} In post-disaster rescue scenarios, the completion delay of tasks could greatly affect the satisfaction of C-UAVs. Similar to~\cite{review-A6,computationmodel2}, a convex logarithmic function is employed to quantify the satisfaction of C-UAVs on task completion. Therefore, the revenue obtained by C-UAV $n$ can be calculated as:
 \begin{equation}
    \label{eq-benefit}
     Bt_n(t) = \log(\beta +T_n^{\text{max}}(t)-T_n(t)),
 \end{equation}
 \noindent where {\color{color} $\beta$ is a constant with a positive value that ensures the revenue function non-negative} and $T_n(t)$ is the completion delay of task $\mathbf{\Phi}_n(t)$.

\textbf{\textit{Cost of Energy Consumption.}} Considering the limited battery capacity of the C-UAVs, the cost of C-UAV $n$ is modeled as the energy consumption, which is given as:
\begin{equation}
    \label{eq-cost-energy}
    Ct_n^{\text{E}}(t) = E_n(t).
\end{equation}
\par \textbf{\textit{Cost of Computation Resource.}} In the MEC-VFC-aided aerial-terrestrial UAV network, the C-UAVs share the computation resources of E-UAV. However, the limited computing resources of the E-UAV and the stringent demands of C-UAVs could lead to resource competition among the C-UAVs and rapid resource depletion of the E-UAV. To ensure the effective utilization and sustainability of resources, the price-based mechanism is introduced to model the cost of using the E-UAV computation resources. Similar to the existing work~\cite{charge1,charge2,charge3}, the cost that C-UAV $n$ pays for the computation resources of E-UAV is given as:
\begin{equation}
    \label{eq-cost-charge}
    Ct_n^{\text{mec}}(t) = \rho_0 F_n(t),
\end{equation}
\noindent where $\rho_0$ represents the unit price of computing resources for E-UAV.

\par According to the above metrics, we finally design the utility function of C-UAV $n$ as follows:
\begin{equation}
U_n^i(t)= 
\begin{cases}U_n^{\text{loc}}(t), & a_n^{\text{loc}}(t)=1\\ 
U_n^{\text{veh}}(t), & a_n^{\text{veh}}(t)=1, \\ 
U_n^{\text{mec}}(t), & a_n^{\text{mec}}(t)=1
\end{cases}
\end{equation}
\noindent where $U_n^{\text{loc}}(t)$ is the utility of local computing, $U_n^{\text{mec}}(t)$ is the utility of MEC-assisted offloading, and $U_n^{\text{veh}}(t)$ is the utility of VFC-assisted offloading, which are denoted respectively as:
\begin{sequation}
	\label{eq.utility}
	\begin{cases}
		U_n^{\text{loc}}(t)=\alpha_n \log(\beta+T_n^{\text{max}}(t)-T_n^{\text{loc}}(t)) - \beta_n E_n^{\text{loc}}(t)\\
	    U_n^{\text{veh}}(t)=\alpha_n \log(\beta+T_n^{\text{max}}(t)-T_n^{\text{veh}}(t)) - \beta_n E_n^{\text{veh}}(t). \\
            U_n^{\text{mec}}(t)=\alpha_n \log(\beta+T_n^{\text{max}}(t)-T_n^{\text{mec}}(t)) - \beta_n E_n^{\text{mec}}(t)\\
            \qquad \qquad - \rho_0 F_n(t)
	\end{cases}
\end{sequation}

\noindent Moreover $\alpha_n$ and $\beta_n$ denote the coefficients of task completion delay and energy consumption, respectively, and $\alpha_n + \beta_n = 1$.

%
%
\subsection{Problem Formulation}
\label{subsec:Problem Formulation}

\par This work aims to maximize the time-average system utility by jointly optimizing the task offloading decisions $\mathcal{A}=\{\mathcal{A}^t|\mathcal{A}^t=\{a_n^i(t)\}_{n\in\mathcal{N},i\in\mathcal{I}}\}_{t\in\mathcal{T}}$, MEC computing resource allocation $\mathcal{F}=\{\mathcal{F}^t|\mathcal{F}^t=\{F_n(t)\}_{n\in\mathbf{N_0}}\}_{t\in\mathcal{T}}$, and VFC resource allocation including vehicle fog node selection $\mathcal{S}=\{\mathcal{S}^t|\mathcal{S}^t=\{\mathbf{S}_n^{'}(t)\}_{n\in\mathbf{N_1}}\}_{t\in\mathcal{T}}$ and task division $\mathbf{\Lambda}=\{\mathbf{\lambda}^t|\mathbf{\lambda}^t=\{\lambda_n^t\}_{n\in\mathbf{N_1}}\}_{t\in\mathcal{T}}$, where $\mathbf{N}_0$ and $\mathbf{N}_1$ denote the set of C-UAVs that choose MEC and VFC at time $t$, respectively. Accordingly, the JTRAOP can be formulated as follows:
\begin{align}
    \textbf{P}:\ &\underset{\mathcal{A},\mathcal{F},\mathcal{S},\mathbf{\Lambda}}{\text{max}} \frac{1}{T}\sum_{t=0}^{T-1}\sum_{i\in \mathcal{I}}\sum_{n\in \mathcal{N}}g_n(t) a_n^i(t) U_n^i(t) \label{P}\\
    \text{s.t.}  \ &a_n^i(t)=\{0,1\}, \forall n\in \mathcal{N}, \forall i\in \mathcal{I}, \forall t\in \mathcal{T} \tag{\ref{P}{\text{a}}} \label{Pa}\\
    \ &\sum_{i\in \mathcal{I}}a_n^i(t)=1,\forall n\in \mathcal{N},\forall t\in \mathcal{T}\tag{\ref{P}{\text{b}}} \label{Pb}\\
        \ &0\leq F_n(t)\leq F_{u}^{\text{max}}, \forall n\in \mathbf{N}_0,\forall t\in \mathcal{T}\tag{\ref{P}{\text{c}}} \label{Pc}\\
        \ &\sum_{n\in \mathbf{N}_0}F_n(t)\leq F_{u}^{\text{max}},\forall t\in \mathcal{T}\tag{\ref{P}{\text{d}}} \label{Pd}\\
        \ &\lambda_{n,j}^t\in [0,1],\forall n\in \mathbf{N}_1,\forall j\in \mathbf{S}^{'}_n(t), \forall t\in \mathcal{T}\tag{\ref{P}{\text{e}}} \label{Pe}\\
        \ &\sum_{j\in \mathbf{S}_n^{'}(t)}\lambda_{n,j}^t=1,\forall n\in \mathbf{N}_1,\forall t\in \mathcal{T}\tag{\ref{P}{\text{f}}} \label{Pf}\\
        \ &\mathbf{S}_n^{\prime}(t)\cap \mathbf{S}_j^{\prime}(t)=\emptyset,\forall n\neq j,\ n,j\in \mathbf{N_1},\forall t\in \mathcal{T},\tag{\ref{P}{\text{g}}} \label{Pg}\\
        \ &\mathbf{S}_n^{'}(t)\subset \mathbf{S}_n(t), \forall n\in \mathbf{N}_1,\forall t\in \mathcal{T}\tag{\ref{P}{\text{h}}} \label{Ph}\\
        \ &g_n(t)=\{0,1\},\forall n\in \mathcal{N},\forall t\in \mathcal{T}\tag{\ref{P}{\text{i}}} \label{Pi}\\
        \ &a_n^i(t)T_n^i(t)\leq T_n^{\text{max}}(t),\forall n\in \mathcal{N}, \forall i\in \mathcal{I}, \forall t\in \mathcal{T}\tag{\ref{P}{\text{j}}} \label{Pj}
\end{align}
\noindent Constraints (\ref{Pa}) and (\ref{Pb}) represent each C-UAV can only choose one offloading strategy. Constraints (\ref{Pc}) and (\ref{Pd}) indicate that the computation resources allocated by the E-UAV should be positive and not greater than the maximum resource owned by the E-UAV. Constraints (\ref{Pe}) and (\ref{Pf}) pose the conditions on task division when C-UAVs decide to offload the tasks to rescue vehicles for execution. Constraint (\ref{Pg}) represents that each vehicle fog node is selected to serve one C-UAV. Constraint (\ref{Ph}) ensures that the selected vehicle fog node should be within the communication range of the C-UAV. Constraint (\ref{Pi}) represents whether a C-UAV generates a computing task at time $t$. Moreover, Constraint (\ref{Pj}) means that the maximum acceptable delay should not be exceeded in completing the task.
\par Similar to~\cite{assume1}, we assume that the tasks generated by the C-UAVs can be completed within one time slot since the computing tasks of rescue missions are delay-sensitive. Therefore, the optimization problem $\textbf{P}$ can be reformulated as the real-time optimization problem $\textbf{P}^{\prime}$ that maximizes the system utility in each time slot, which is given as:
\begin{align}
    \textbf{P}^{\prime}:\ &\underset{\mathcal{A}^t,\mathcal{F}^t,\mathcal{S}^t,\lambda^t}{\text{max}} \sum_{i\in I}\sum_{n\in N}g_n(t) a_n^i(t) U_n^i(t) \label{opti1}\\
    \text{s.t.}  \ &(\ref{Pa})-(\ref{Pi}) \notag
\end{align}

\noindent where $\{\mathcal{A}^t,\mathcal{F}^t,\mathcal{S}^t,\lambda^t\}$ indicates the decisions of task offloading, MEC computation resource allocation, vehicle fog node selection, and task division at time slot $t$. The above problem $\textbf{P}^{'}$ contains both binary variables (i.e., task offloading decision $\mathcal{A}^t$ and vehicle fog node selection $\mathcal{S}^t$) and continuous variables (i.e., MEC computation resource allocation $\mathcal{F}^t$ and task division $\lambda^t$) is a mixed-integer non-linear programming (MINLP) problem, which is non-convex~\cite{boyd2004convex,Guo2019} and NP-hard~\cite{belotti2013mixed}. Therefore, a large amount of computational overhead caused by seeking the optimal solution may not be suitable for real-time decision making. To this end, we design an MVTORA approach that obtains a sub-optimal solution in polynomial time complexity. Furthermore, for the convenience of the following description, we drop the time index for variables similar to~\cite{Cui2023}.

%
%
\section{MEC-VFC-Aided Task Offloading and Resource Allocation Approach}
\label{sec:MVTORA}
\par To achieve the maximal system utility, the MVTORA approach is presented by separating problem $\textbf{P}^{'}$ into two parts, i.e., task offloading and computing resource allocation, which are solved respectively. First, the task offloading part seeks to optimize the task offloading decisions for C-UAVs, which is solved by adopting game theory. Furthermore, the resource allocation part aims to optimize the MEC and VFC resource allocation decisions for the E-UAV and vehicular fog nodes, respectively, which are solved by employing convex optimization and evolutionary computation, respectively. The task offloading and computing resource allocation are detailed in Sections \ref{subsec:Computation Offloading Decisions} and \ref{subsec:The Optimal MEC Computing Resource Allocation}, respectively. In addition, Section \ref{subsec:Complexity Analysis} presents the main steps and analysis of the MVTORA approach. {\color{b} Note that we employ a binary offloading strategy for MEC offloading and a partial offloading strategy for VFC offloading since the computing capability of the E-UAV is powerful while that of the ground rescue vehicles is relatively limited. A comprehensive explanation of this offloading decision is presented in Section \ref{sec:Discussion}.}

\subsection{Task Offloading}
\label{subsec:Computation Offloading Decisions}

\par The offloading decision of C-UAV $n$ depends not just on its own demand but also on the offloading decisions of the other C-UAVs. Considering the competitive nature of task offloading among C-UAVs, game theory is employed to solve the task offloading decision problem.

%
%
\subsubsection{Game Formulation} 
\label{subsubsec:Game Formulation}

\par The problem of task offloading decision is modeled as a task offloading game among multiple C-UAVs, which is defined as a triplet $\Gamma=\{\mathcal{N},\mathbb{A},(U_n)_{n\in \mathcal{N}}\}$, where the elements are detailed as follows:

\begin{itemize}
\item $\mathcal{N}=\{1,2,\dots,N\}$ denotes the players, i.e., all C-UAVs.
\item $\mathbb{A}=\mathbf{A}_1\times\dots\times\mathbf{A}_N$ denotes the strategy space, wherein $\mathbf{A}_n=\{a_n^{\text{loc}},a_n^{\text{mec}},a_n^{\text{veh}}\}$ is the set of offloading strategies for player $n\ (n\in \mathcal{N})$, $a_n\in\mathbf{A}_n$ denotes the strategy chosen by player $n$, and $\mathcal{A}=(a_1,\dots,a_N)\in \mathbb{A}$ is the strategy profile.
\item $(U_n)_{n\in \mathcal{N}}$ is the utility function of player $n$ that maps each strategy profile $\mathcal{A}$ to a real number, i.e., $U_n(\mathcal{A}): \mathbb{A} \mapsto \mathbb{R}$, where $\mathbb{R}$ is the set of real number.
\end{itemize}
\noindent Each C-UAV aims to maximize its utility by choosing an optimal offloading strategy. Thus, the problem of task offloading can be formulated as:
\begin{align}
    \underset{a_n}{\text{max}} & U_n(a_n,a_{-n})=a_n^{\text{loc}}U_n^{\text{loc}}+a_n^{\text{veh}}U_n^{\text{veh}}+a_n^{\text{mec}}U_n^{\text{mec}}\label{P1}\\
        \text{s.t.}  \ &a_n^{\text{loc}}+a_n^{\text{veh}}+a_n^{\text{mec}}=1,\forall n\in \mathcal{N},\tag{\ref{P1}{\text{a}}} \label{P1a}\\
        \ &a_n^i=\{0,1\}, \forall n\in \mathcal{N}, i\in \{\text{loc}, \text{mec}, \text{veh}\},\tag{\ref{P1}{\text{b}}} \label{P1b}
\end{align}
\noindent where $a_{-n}=(a_1,\dots,a_{n-1},a_{n+1},\dots,a_N)$ denotes the offloading decisions of the other players except player $n$.
%
%
\subsubsection{The Solution to Task Offloading Game}
\label{subsubsec:The Existence of Nash Equilibrium}
\par To determine the solution to the task offloading game, we first introduce the concept of Nash equilibrium, which describes a situation where no player has any incentive to unilaterally deviate from the current strategy.
\begin{definition}
\label{def:def1}
    The strategy profile $\mathcal{A}^*=(a_1^*,\dots,a_N^*)$ is a pure-strategy Nash equilibrium of game $\Gamma$ if and only if
    \begin{equation}
    U_n(a_n^*,a_{-n}^*)\geq U_n(a_n^{\prime},a_{-n}^*) \quad  \forall a_n^{\prime}\in \mathbf{A}_n,\forall n \in \mathcal{N}.
    \end{equation}
\end{definition}
\par {\color{color1} Furthermore, we introduce a powerful tool, known as exact potential games~\cite{potential}, to help us study the existence of Nash equilibrium and how to obtain a Nash equilibrium solution for the task offloading game.}
\begin{definition}
    \label{def:def2}
    {\color{color} A game can be called an exact potential game if and only if a potential function $F(\mathcal{A}): \mathbb{A} \mapsto \mathbb{R}$ exists such that} 
    \begin{equation}
    \label{PG-def}
    \begin{split}
        &U_n(a_n,a_{-n})-U_n(b_n,a_{-n})\\&= F(a_n,a_{-n})-F(b_n,a_{-n}), \forall n\in \mathcal{N}, \textcolor{color1}{}a_n,b_n\in \mathbf{A}_n,   
    \end{split}
    \end{equation}
\noindent where $F(\mathcal{A})$ accurately captures the utility change of a single player due to strategic deviation. 
\end{definition}

{\color{b}
\par Besides, we introduce how to obtain a Nash equilibrium solution of the exact potential game by presenting the concepts of the finite improvement path (FIP) and the better response update process.
}

\begin{definition}
    The exact potential game with finite strategy sets always has a Nash equilibrium and the FIP~\cite{potential}. 
    \label{def:def_FIP}
\end{definition}

{\color{b}
\begin{definition}
    In the better response update process, given the other players’ strategy $a_{-n}$, player $n$ will select a new strategy $T_n$ over the current strategy $a_n$ if and only if $T_n$ is any randomly selected strategy that improves his/her utility. We formally write it as 
    \begin{equation}
    \begin{aligned}
    &T_n=\operatorname{rand}\left(\left\{a_n^{\prime} \mid U_n\left(a_n^{\prime}, a_{-n}\right)>U_n\left(a_n, a_{-n}\right)\right\}\right),\\ 
    &\forall a_n^{\prime}\in \mathbf{A}_n, n \in \mathcal{N},
    \end{aligned}
     \end{equation}
    where $\operatorname{rand}(\{.\})$ denotes a randomized selection among elements of a set.
    \label{def:def_beteer_response}
\end{definition}}

{\color{b}
\par According to Definitions \ref{def:def_FIP} and \ref{def:def_beteer_response}, the FIP means that each player updates its current strategy in each iteration through the better response update process and after a finite number of iterations, the improvement path terminates and its end point corresponds to the Nash equilibrium solution~\cite{la2016potential}. Therefore, for an exact potential game, we can obtain the Nash equilibrium solution by the better response update process. 

\par Finally, we prove that the task offloading game among multiple C-UAVs is an exact potential game through the following Theorem \ref{lem:lem1}.}
\begin{theorem}
\par The task offloading game among multiple C-UAVs is an exact potential game where the potential function $F(\mathcal{A})$ is given as:
\begin{equation}
\label{potential-function}
\begin{aligned}
F(\mathcal{A}) 
& =a_n^{\mathrm{loc}} \sum_{j=1}^N\left(\alpha_j\log \left(\beta+T_j^{\mathrm {max}}-T_j^{\mathrm{loc}}\right) -\beta_jE_j^{\mathrm{loc}}\right)\\
& +\left(1-a_n^{\mathrm{loc}}\right)\times\left\{\alpha_n \log \left(\beta+T_n^{\mathrm{max}}-a_n^{\mathrm{mec}} T_n^{\mathrm{mec}}\right.\right.\\
&\left.-a_n^{\mathrm{veh}} T_n^{\mathrm{veh}}\right)\left.-\beta_n\left(a_n^{\mathrm {mec}} E_n^{\mathrm {mec}}+a_n^{\mathrm{veh}} E_n^{\mathrm{veh}}\right)-\rho_0 a_n^{\mathrm {mec}}F_n\right.\\
& +\sum_{j=1, j \neq n}^N\left(\alpha_j\log \left(\beta+T_j^{\mathrm {max}}-T_j^{\mathrm{loc}}\right)-\beta_jE_j^{\mathrm{loc}}\right)\}.
\end{aligned}
\end{equation}
\label{lem:lem1}
\end{theorem}
\begin{proof}
{\color{b} The detailed proof is given in Appendix A of the supplemental material.}
\label{pro:pro1}
\end{proof}

\par The key idea of the task offloading game is to iteratively update the players' offloading strategies through the better response update process until the Nash equilibrium is reached, which is shown in Algorithm \ref{Algorithm 1}. The main steps of implementing the task offloading game are described as follows.
\textbf{i)} In each time slot, the E-UAV collects the state information of C-UAVs, the CSI of U2U channel, and the initial task offloading decision and corresponding utility of C-UAVs. \textbf{ii)} Each iteration is divided into $N$ decision slots (Lines 5$\sim$10). At each decision slot, one C-UAV is selected to attempt to update its offloading decision while the offloading decisions of other C-UAVs remain unchanged (Line 6). \textbf{iii)} If higher utility is achieved, the C-UAV's offloading decision is updated; otherwise the original offloading decision is maintained (Lines 7$\sim$10). \textbf{iv)} When no C-UAV changes its offloading decision, the task offloading decision game reaches the Nash equilibrium. \textbf{v)} The E-UAV sends the optimal task offloading decision information to each C-UAV. \textbf{vi)} The C-UAVs perform the actions of computation offloading according to the received decisions. 
\begin{algorithm}[]	
    \label{Algorithm 1}
    \SetAlgoLined
    \KwIn{The state information of C-UAVs $\{\mathbf{St_n^{\text{uav}}}\}_{n\in \mathcal{N}}$, the initial task offloading decision $\mathbf{A}^{\text{ini}}=\{a_n\}_{n\in\mathcal{N}}$ and corresponding utility $\mathbf{U}^{\text{ini}}=\{U_n\}_{n\in\mathcal{N}}$.}
    \KwOut{The optimal task offloading decision $\mathbf{A}^{*}=\{a_n^{*}\}_{n\in\mathcal{N}}$.}
    \textbf{ Initialization:} 
    Iteration $l=1$, $\mathbf{A}^0=\emptyset$\;
    $\mathbf{A}^l=\mathbf{A^{\text{ini}}}$\;
    \While{$\mathbf{A}^{l-1}\neq \mathbf{A}^l$}
    {
        $\mathbf{A}^{l-1}=\mathbf{A}^{l}$\;
        \For{$n\in \mathcal{N}$}
        {
            $\mathbf{A}^{l}(n)=a_n^{\text{mec}}=1$\;
            Call Algorithm \ref{Algorithm 2} for $F_n^{*}$ based on $\mathbf{A}^{l}$\;
            Calculate the utility $U_n^{\text{mec}}$ based on $F_n^{*}$ and Eq. (\ref{eq.utility})\;
            \If{$U_n^{\text{mec}}\leq\mathbf{U}^{\text{ini}}(n)$}
            {
                $\mathbf{A}^{l}(n)=\mathbf{A}^{\text{ini}}(n)$\;
            }
            
        }
        $l=l+1$\;
    }
    $\mathbf{A}^{*}=\mathbf{A}^{l}$\;
    \Return{$\mathbf{A}^{*}=\{a_n^{*}\}_{n\in\mathcal{N}}$.}
    \caption{Task Offloading Game}
\end{algorithm}	
%
%
\subsection{Resource Allocation}
\label{subsec:The Optimal MEC Computing Resource Allocation}

\par The problem of resource allocation is decomposed into the sub-problems of MEC resource allocation and VFC resource allocation, respectively, which aim to obtain the optimal resource allocation decisions for aerial E-UAV and terrestrial vehicle nodes, respectively. 

\subsubsection{MEC Resource Allocation} The MEC resource allocation problem $\textbf{P1}$ seeks to maximize the total utility of C-UAVs that offload tasks to the aerial E-UAV by optimizing the resource allocation of E-UAV, which is formulated as: 
\begin{align}
    \textbf{P1}: \ \underset{\mathcal{F}}{\text{max}} &\sum_{n\in \mathbf{N_0}}\left\{\alpha_n \log(\beta+T_n^{\text{max}}-T_n^{\text{mec}})-\beta_nE_n^{\text{mec}}-\rho_0F_n \right\}\label{P2}\\
    \text{s.t.}  \ &0\leq F_n\leq F_u^{\text{max}}, \forall n\in \mathbf{N_0},\tag{\ref{P2}{\text{a}}}\label{P2a}\\
        \ &\sum_{n\in \mathbf{N_0}}F_n\leq F_u^{\text{max}}. \tag{\ref{P2}{\text{b}}}\label{P2b}
\end{align}
\begin{lemma}
\label{lem:lem2}
Problem $\textbf{P1}$ is convex. 
\end{lemma}
\begin{proof}
\label{pro:pro2}
{\color{b} The detailed proof is given in Appendix B of the supplemental material.}
\end{proof}
\begin{theorem}
The solution to Problem $\textbf{P1}$, i.e., the optimal computation resource allocated by the E-UAV to the C-UAVs, is given as $\mathcal{F}^{*}=\{F_n^{*}, n\in \mathbf{N_0}\}$, where
\begin{sequation}
    \label{eq-optimal_Fi}
    F_n^{*}=\frac{\eta_n D_n+\sqrt{(\eta_n D_n)^2-4\left(\beta+T_n^{\mathrm{max}}-\frac{D_n}{R_{n,u}}\right)\left(-\frac{\eta_n D_n \alpha_n}{\rho_0+\gamma^{*}}\right)}}{2\left(\beta+T_n^{\mathrm{max}}-\frac{D_n}{R_{n,u}}\right)}. 
\end{sequation}
\end{theorem}
\begin{proof}
 Since Problem $\textbf{P1}$ is convex and the slater condition is satisfied. Hence, we can solve Problem $\textbf{P1}$ by using the partial Lagrange function, which is given as
\begin{equation}
    \begin{aligned}
        L(\mathcal{F},\gamma)& =\sum_{n\in \mathbf{N_0}}\{\alpha_n \log(\beta+T_n^{\mathrm{max}}-T_n^{\mathrm{mec}})-\beta_nE_n^{\mathrm{mec}} \\
        &-\rho_0F_n\} -\gamma(\sum_{n\in \mathbf{N_0}}F_n-F_u^{\mathrm{max}}),
        \label{L}
    \end{aligned}
\end{equation}
\noindent where $\gamma\geq 0$ is the Lagrange multiplier associated with the computation resource constraint of E-UAV (\ref{P2b}). Subsequently, the  Karush–Kuhn–Tucker (KKT) conditions are employed to take the optimal computation resource allocation $\mathcal{F}$. By using the first-order optimality condition, Eq. \eqref{eq-optimal_Fi} can be achieved.
\end {proof}

\par As shown in Algorithm \ref{Algorithm 2}, the optimal MEC resource allocation can be achieved by applying the bisection method~\cite{computationmodel2}.
\begin{algorithm}[]	
    \label{Algorithm 2}
    \SetAlgoLined
    \KwIn{Task set $\{\mathbf{\Phi}_n\}_{n\in\mathbf{N_0}}$, E-UAV computation resources $F_u^{\text{max}}$.}
    \KwOut{The optimal computation resource allocation $\mathcal{F}^{*}=\{F_n^{*}, n\in \mathbf{N}_0\}$.}
    \textbf{ Initialization:} 
	Search accuracy threshold: $\varepsilon$, the lower bound $\gamma^{\text{min}}=0$ and the upper bound $\gamma^{\text{max}}=\gamma^{\text{bound}}$\;
    \While {$\gamma^{\text{max}}-\gamma^{\text{min}}\geq \varepsilon$}
    {
        Define $\gamma=\frac{\gamma^{\text{min}}+\gamma^{\text{max}}}{2}$\;
        \For {$n\in \mathbf{N}_0$}
        {
            Compute $F_n^{*}$ by substituting $\gamma$ into Eq. (\ref{eq-optimal_Fi});
        }
        \eIf {$\sum_{n\in \mathbf{N}_0}F_n^{*}\geq F_u^{\text{max}}$}
        {
            $\gamma^{\text{min}}=\gamma$\;
        }
        {
            $\gamma^{\text{max}}=\gamma$\;
        }
    }
    \Return{$\mathcal{F}^{*}=\{F_n^{*}, n\in \mathbf{N}_0\}$.}
    \caption{Bisection Algorithm-based MEC Resource Allocation}
\end{algorithm}	

%
%
\subsubsection{{VFC Resource Allocation}} 

\par The VFC resource allocation problem $\textbf{P2}$ aims to maximize the total utility of C-UAVs that offload tasks to terrestrial vehicles by optimizing the resource allocation of vehicle fog nodes. Since the task of each C-UAV is divided into multiple independent sub-tasks and offloaded to a set of vehicle fog nodes for parallel processing, as explained in Section \ref{subsec:Computation Model}, Problem $\textbf{P2}$ is solved by mapping the VFC resource allocation into the vehicle fog node selection and task division, which is formulated as: 
\begin{align}
\textbf{P2}:\ &\underset{\mathbf{\lambda},\mathcal{S}}{\text{max}} \sum_{n\in \mathbf{N}_1}\{\alpha_n \log(\beta +(T_n^{\text{max}}-T_n^{\text{veh}}))-\beta_nE_n^{\text{veh}}\}\label{P3}\\
    \text{s.t.}  \ &\lambda_{n,j}\in [0,1],\forall n\in \mathbf{N}_1, \forall j\in \mathbf{S}_n^{\prime},\tag{\ref{P3}{\text{a}}}\label{P3a}\\
    \ &\sum_{j\in \mathbf{S}_n^{'}}\lambda_{n,j}=1,\forall n\in \mathbf{N}_1,\tag{\ref{P3}{\text{b}}}\label{P3b}\\
    \ &\mathbf{S}_n^{'}\cap \mathbf{S}_j^{'}=\emptyset,\forall n,j\in \mathbf{N}_1,n\neq j,\tag{\ref{P3}{\text{c}}}\label{P3c}\\
    \ &\mathbf{S}_n^{'}\subset \mathbf{S}_n, \forall n\in \mathbf{N}_1.\tag{\ref{P3}{\text{d}}}\label{P3d}
\end{align}
\par Since the communication ranges of C-UAVs do not overlap each other as mentioned in Section \ref{subsec:System Overview}, the selection of vehicle fog nodes for each C-UAV is independent of each other. Therefore, \textbf{P2} can be decomposed into $|\mathbf{N}_1|$ parallel sub-problems, where each sub-problem is expressed as:
\begin{align}
\textbf{P2}^{\prime}:\  & \underset{\lambda_n,\mathbf{S}_n^{\prime}}{\text{max}}\{\alpha_n \log(\beta +T_n^{\text{max}}-T_n^{\text{veh}})-\beta_nE_n^{\text{veh}}\} \label{P3.1}\\
    \text{s.t.}  \ &(\ref{P3a})-(\ref{P3d}). \notag
\end{align}
\par Problem $\textbf{P2}^{\prime}$ is still an MINLP problem, which is difficult to be solved directly. Since the solutions of vehicle fog node selection $\mathbf{S}_n^{\prime}$ and task division $\lambda_n$ are inherently sequential, i.e. the vehicle fog node selection is performed before the task division, this inspires us to solve the problem by designing a two-step optimization procedure which includes the vehicle fog node selection and task division, and the details are as follows.

\par \textbf{(1) Vehicle Fog Node Selection.} Since the mission-critical computing tasks generated by C-UAVs are heterogeneous and delay-sensitive, the vehicle fog nodes are selected according to the different preferences of C-UAVs, with the aim of minimizing the task completion delay. Therefore, we define the preference value of C-UAV $n$ to vehicle $j$ as
\begin{equation}
    \label{eq-preference}
    Pr(n,j) = \frac{D_n}{R_{n,j}} + \frac{\eta_n D_n}{f_j^{\text{veh}}}.
\end{equation}
\par The vehicle fog nodes can be selected based on the following rule.
\begin{theorem}
If the vehicle set $\mathbf{S}_n$ is sorted in increasing order of preference, the top $K_n$ vehicles are the optimal candidate vehicle set $\mathbf{S}_n^{*}$ which minimizes the completion delay of task $\mathbf{\Phi}_n$.
\label{lem:lem3}
\end{theorem}
\begin{proof}
We assume that the set of vehicles within the communication range of C-UAV $n$ are sorted by increasing preference as $\mathbf{S}_n=\{m_1, \ldots, m_{|\mathbf{S}_n|}\}$. Therefore, the optimal candidate vehicle set is $\mathbf{S}_n^{*}=\{m_1, m_2, \ldots, m_{K_n}\}$. Suppose $\lambda_n=\{\lambda_{n, 1}, \ldots, \lambda_{n,K_n}\}$ is an arbitrary set of task partitions for C-UAV $n$, then the completion delay $T_n(\mathbf{S}_n^{*},\lambda_n)$ of task $\mathbf{\Phi}_n$ can be expressed as:
\begin{equation}
\label{eq.comp}
\begin{aligned}
   T_n(\mathbf{S}_n^{*},\lambda_n)&=\max\limits_{j\in \mathbf{S}_n^{*}}\left(\frac{\lambda_{n,j} D_n}{R_{n,j}}+\frac{\lambda_{n,j} \eta_nD_n}{f_j^{\text{veh}}}\right)\\
   &=\max\limits_{j\in \mathbf{S}_n^{*}}\left(\lambda_{n,j}Pr(n,j)\right).
\end{aligned}
\end{equation}
\noindent For any set of candidate vehicles $\mathbf{S}_n^{\prime}=\{z_1, z_2, \ldots, z_{K_n}\}$ sorted in ascending order of preference, the following formula holds.
\begin{equation}
\begin{split}
    &\lambda_{n, j} Pr(n,m_j)\leq  \lambda_{n, j} Pr(n,z_j),\ j\in \{1,2,\ldots,K_n\}.
\end{split}
\end{equation}
\noindent Therefore, $T_n(\mathbf{S}_n^{*},\lambda_n)\leq T_n(\mathbf{S}_n^{\prime},\lambda_n)$, $\mathbf{S}_n^{*}$ is the optimal candidate vehicle set.
\label{pro:pro3}
\end{proof}
\par Based on the Theorem \ref{lem:lem3}, the method of vehicle fog node selection is shown in Algorithm \ref{Algorithm 3}.
\begin{algorithm}
    \label{Algorithm 3}
    \SetAlgoLined
    \KwIn{Task $\mathbf{\Phi}_n(t)$ and the vehicle set $\mathbf{S}_n$.}
    \KwOut{The optimal candidate vehicle set $\mathbf{S}_n^{*}$.}
    \eIf {$|\mathbf{S}_n|\leq K_n$}
    {
        $\mathbf{S}_n^{*} = \mathbf{S}_n$;
    }
    {
        Calculate the preference value of all vehicles in set $\mathbf{S}_n$ based on Eq. \eqref{eq-preference}\;
        Select the top $K_n$ vehicles with the smallest preference value as the optimal candidate vehicle set $\mathbf{S}_n^{*}$\;
    }
    \Return{$\mathbf{S}_n^{*}$.}
    \caption{Vehicle Fog Node Selection}
\end{algorithm}
%
%


\par \textbf{(2) Task Division.} Given the selection of vehicle fog nodes $\mathbf{S}_n^{*}$, problem $\textbf{P2}^{\prime}$ can be transformed into a task division problem, which is expressed as follows:
\begin{align}
    \textbf{P2}^{\prime \prime}:\  & \underset{\lambda_n}{\text{max}}\{\alpha_n \log(\beta +T_n^{\text{max}}-T_n^{\text{veh}})-\beta_nE_n^{\text{veh}}\}\label{P3.2}\\
        \text{s.t.}  \ &(\ref{P3a})-(\ref{P3b}). \notag   
\end{align}
\par The service delay $T_n^{\text{veh}}$ of VFC-assisted task offloading is a maximum function given in Eq.~\eqref{eq.veh-delay}, which makes the problem $\textbf{P2}^{\prime \prime}$ nondifferentiable. Therefore, it is difficult to directly solve the problem $\textbf{P2}^{\prime \prime}$. Algorithms based on evolutionary computation have the potential to solve this problem, which does not require convexity and differentiability of the optimization problem. To this end, we design a task division algorithm by employing genetic algorithm (GA) because of its global search ability, parallel processing capability, and strong robustness. {\color{color} Moreover, since the problem $\textbf{P2}^{\prime \prime}$ has a small-scale solution space (i.e., $|\lambda_n|\leq K_n$), the running time of the algorithm can be guaranteed for real-time decision-making.}

\par In particular, GA inspires from biological evolution process~\cite{GA2}, in which a population with size $L$ is first initialized, and each individual in the population represents a potential solution to the optimization problem. Then, the fitness of each individual in the population is evaluated based on the objective function (\ref{P3.2}), and $L$ parents are chosen from the population according to the fitness of these individuals. Moreover, $L$ parents produce $L$ offspring through crossover operation, and $L$ offspring mutate with a certain probability to form the next generation population. Over successive population iterations, the optimal or the feasible sub-optimal solution is obtained. Different from the unconstrained optimization problem, Problem $\textbf{P2}^{\prime \prime}$ is restricted by the equality constraint (i.e., $(\ref{P3b})$). However, the traditional GA cannot directly solve the constrained optimization problems~\cite{penalty}. Therefore, we first handle the constraint with the following additional operation.

\par To satisfy equality constraint (\ref{P3b}), after each generation population is formed, each individual $X_l=\{x_{l1},x_{l2},\dots,x_{lK}\}$ ($K=|\lambda_n|$) is normalized as
\begin{equation}
	x_{lj}=\frac{x_{lj}}{\sum_{k=1}^{K}x_{lk}}, j\in \{1,2,\dots,K\}. \label{eq.linear-constraint}
\end{equation}
\par The task division algorithm is shown in Algorithm \ref{Algorithm 4} and the specific genetic operators are given as follows:
\begin{algorithm}
    \label{Algorithm 4}
    \SetAlgoLined
    \KwIn{Task $\mathbf{\Phi}_n(t)$, vehicle set $\mathbf{S}_n^{*}$, maximum evolution generation $G$, population size $L$, crossover probability $pc$ and mutation probability $pm$.} 
    \KwOut{The optimal task division set $\lambda_n^{*}$.}
    \tcp{Initialize the population}
    \For{$l=1$ to $L$}
    {
        Initialize the $l$th individual of the population through the initialization operation\;
        Normalize the individual based on Eq. \eqref{eq.linear-constraint}\;
    }
    \For{$g=1$ to $G$}
    {
        Calculate the fitness of each individual in the population based on (\ref{P3.2})\;
        Select the elite individual $X^{*}$ with the highest fitness in the population\;
        $\lambda_n^{*}=X^{*}$\;
        Select the parent population through the selection operation\;
        Obtain the offspring population through the crossover operation\;
        Mutate the offspring population through the mutation operation\;
        Normalize the individual based on Eq. \eqref{eq.linear-constraint}\;
        Replace the lowest fitness individual in the offspring population with the elite individual\;
    }
    \Return{$\lambda_n^{*}$.}
    \caption{Task Division}
\end{algorithm}

\par \textbf{Initialization:} In this phase, the initial population is generated by using a real-coding scheme to randomly create $L$ individuals. Each individual $X_l=\{x_{l1},x_{l2},\dots,x_{lK}\}$ $(l\in \{1,2,\dots,L\})$ represents a potential solution of the optimization problem, which is called a chromosome containing $K$ genes. The value of each gene $x_{l,k}$ $(k\in \{1,2,\dots,K\})$ is generated by a random number generator within the range defined by constraint (\ref{P3a}). Specifically, each gene is generated as $x_{l,k}=X^{\text{rand}}$, where $X^{\text{rand}}$ is a uniformly distributed random value within the interval (0, 1).

\par \textbf{Selection:} The elite-reserved 2-tournament selection strategy is employed in this stage, which has the advantages of efficiency and simplicity~\cite{review-A7}. Specifically, two individuals are randomly selected from the population each time, and the individual with higher fitness is chosen as the parent. Then, a parent population is formed until the number of parents reaches $L$. Moreover, the individual with the highest fitness value in the population is selected as the elite individual, which will be used to replace the individual with the lowest fitness in the offspring population. 

\par \textbf{Crossover:} {\color{color1} New offspring are produced by crossing over the genes of the parents.} Specifically, a pair of parents are randomly selected each time from the parent population, and a random number $rand_1\in (0,1)$ is generated at the same time. If $rand_1$ is less than the crossover probability $pc$, two offspring are created by crossing the two parents. Otherwise, the pair of parent individuals does not participate in crossover and are directly copied as offspring. This process continues until an offspring population of size $L$ is obtained. In this work, two offspring (i.e., $\widetilde{X_1}$ and $\widetilde{X_2}$) are produced by a linear combination of the two parents (i.e., $X_1$ and $X_2$). The crossover operation is described as follows:
\begin{equation}
	\label{eq.crossover}
	\begin{cases}
		\widetilde{X_1} = \tau X_1 + (1-\tau)X_2,\\
		\widetilde{X_2} = \tau X_2 + (1-\tau)X_1,
	\end{cases}
\end{equation}
\noindent where $\tau$ is a random number within interval (0, 1).

\par \textbf{Mutation}: The mutation operation acting on genes helps to improve the diversity of individuals. For each gene of each individual in the offspring population, a random number $rand_2\in (0,1)$ is generated to determine whether the gene is mutated. If $rand_2$ is less than the mutation probability $pm$, the gene is {\color{color1} mutated}. Otherwise, the gene remains unchanged. When individual $X_l=\{x_{l1},x_{l2},\dots,x_{lK}\}$ mutates into new individual $\widetilde{X_l}=\{x_{l1},x_{l2},\dots,\widetilde{x_{lj}},\dots,\widetilde{x_{lk}},\dots,x_{lK}\}$, new genes $\widetilde{x_{lj}}$ and $\widetilde{x_{lk}}$ can be expressed as follows: 
\begin{equation}
	\begin{cases}
		\widetilde{x_{lj}}=X^{\text{rand}},\\
		\widetilde{x_{lk}}=X^{\text{rand}}.
	\end{cases}
\end{equation}

%
%
\subsection{Main Steps of MVTORA and Analysis}
\label{subsec:Complexity Analysis}
\begin{algorithm}
    \label{Algorithm 5}
    \SetAlgoLined
    \KwIn{The state information of the E-UAV, rescue vehicles, and C-UAVs $\{\mathbf{St}^{u},\mathbf{St}^{\text{veh}},\mathbf{St}^{\text{uav}}\}$.} 
    \KwOut{Time-average system utility $TSU$.}
    \textbf{Initialization:} Initialize $TSU=0$\;
    \For{$t=0$ to $T-1$}
    {
        \For{each C-UAV $n\in \mathcal{N}$}
        {
            Obtain the vehicle set $\mathbf{S}_n(t)$ and vehicle state information $\{\mathbf{St}_m^{\text{veh}}(t)\}_{m\in \mathbf{S}_n(t)}$\;
            Calculate the utility of local computing $U_n^{\text{loc}}(t)$ based on Eq. \eqref{eq.utility}\;
            Call Algorithm \ref{Algorithm 3} and Algorithm \ref{Algorithm 4} to obtain $\mathbf{S}_n^{*}$ and $\lambda_n^{*}$\;
            Calculate the utility $U_n^{\text{veh}}(t)$ based on $\mathbf{S}_n^{*}$, $\lambda_n^{*}$ and Eq. \eqref{eq.utility}\;
            \eIf{$U_n^{\text{veh}}(t)>U_n^{\text{loc}}(t)$}
            {
                $a_n(t)=a_n^{\text{veh}}(t)=1$\;
                $U_n(t)=U_n^{\text{veh}}(t)$\;
            }
            {
                $a_n(t)=a_n^{\text{loc}}(t)=1$\;
                $U_n(t)=U_n^{\text{loc}}(t)$\;
            }
        }
        E-UAV obtains the initial task offloading decision $\mathbf{A}^{\text{ini}}(t)=\{a_n(t)\}_{n\in \mathcal{N}}$ and corresponding utility $\mathbf{U}^{\text{ini}}(t)=\{U_n(t)\}_{n\in \mathcal{N}}$ of all C-UAVs\;
        E-UAV calls Algorithm \ref{Algorithm 1} and Algorithm \ref{Algorithm 2} to obtain $\mathbf{A}^{*}(t)$ and $\mathcal{F}^{*}(t)$ based on $\mathbf{A}^{\text{ini}}(t)$ and $\mathbf{U}^{\text{ini}}(t)$\;
        All C-UAVs perform their tasks based on $\mathbf{A}^{*}(t)$ and obtain corresponding utility $U_n^{*}(t)$\;
        System utility $SU(t)=\sum_{n=1}^NU_n^{*}(t)$\;
        $TSU = TSU+SU(t)$\;
    }
    $TSU = TSU/T$\;
    \Return{$TSU$.}
    \caption{MVTORA}
\end{algorithm}
\par The main steps of MVTORA are described in Algorithm \ref{Algorithm 5}, and the corresponding performance and complexity analysis is presented as follows.
\subsubsection{Performance Analysis}
\label{subsubsec:Performance Analysis}
In general, there may be more than one Nash equilibrium in the task offloading game. {\color{b}However, computing the best Nash equilibrium has been proven to be an NP-hard problem~\cite{AckermannRV08,FabrikantPT04}. Therefore, a large amount of computational overhead incurred by seeking the best Nash equilibrium may not be suitable for real-time decision making in the considered post-disaster rescue scenario. To evaluate the performance of the Nash equilibrium solution, the price of anarchy (PoA)~\cite{POA} is introduced to quantify the gap between the worst-case Nash equilibrium and the centralized optimal solutions, which can provide a bound on the sub-optimality of our proposed algorithm.} Let $\Upsilon$ denote the set of Nash equilibrium of the task offloading game, $\mathcal{A}=(a_1,\dots,a_N)$ denote a strategy profile, and $\Tilde{\mathcal{A}}=(\Tilde{a_1},\dots,\Tilde{a_N})$ denote the centralized optimal solution that maximizes the system utility, i.e., $\Tilde{\mathcal{A}}=\arg \max _{\mathcal{A} \in \mathbb{A}} \sum_{n \in \mathcal{N}} U_n(\mathcal{A})$. Then the PoA can be given as:
\begin{equation}
\label{eq.def_poa}
    \mathrm{PoA}=\frac{\min _{\mathcal{A} \in \Upsilon} \sum_{n \in \mathcal{N}} U_n(\mathcal{A})}{\sum_{n \in \mathcal{N}} U_n(\Tilde{\mathcal{A}})}.
\end{equation}
\noindent For the metric of system utility, a larger PoA indicates better performance of the task offloading game solution. 
The following Theorem~\ref{theorem4} analyzes the result.
\begin{theorem}
\label{theorem4}
For the task offloading game among multiple C-UAVs, the PoA defined in Eq.~\eqref{eq.def_poa} satisfies:
\begin{equation}
    \label{eq.POA}
    \frac{\sum_{n=1}^N \max \left\{U_n^{\mathrm{loc}}, U_n^{\mathrm{veh}}\right\}}{\sum_{n=1}^N \max \left\{U_n^{\mathrm{loc}}, U_n^{\mathrm{veh}}, U_{n,\mathrm{max}}^{\mathrm{mec}}\right.\}}\leq \mathrm{PoA} \leq 1.
\end{equation}
\noindent where $U_{n,\mathrm{max}}^{\mathrm{mec}}=\alpha_n \log(\beta+T_n^{\mathrm{max}}-\frac{D_n}{R_{n,u}}-\frac{\eta_n D_n}{\hat{F_n}}) - \beta_n E_n^{\mathrm{mec}} - \rho_0 \hat{F_n}$ and $\hat{F_n}=\min \left\{\frac{\eta_n D_n+\sqrt{(\eta_n D_n)^2-4\left(\beta+T_n^{\mathrm{max}}-\frac{D_n}{R_{n,u}}\right)\left(-\frac{\eta_n D_n \alpha_n}{\rho_0}\right)}}{2\left(\beta+T_n^{\mathrm{max}}-\frac{D_n}{R_{n,u}}\right)},F_u^{\mathrm{max}}\right\}$.
\end{theorem}
\begin{proof}
{\color{b} The detailed proof is given in Appendix C of the supplemental material.}
\end{proof}
\subsubsection{Complexity Analysis}
\label{subsubsec:Complexity Analysis}
\begin{theorem}
 MVTORA has a polynomial computational complexity in each time slot, i.e.,  $O\left(I_c N\log _2\left(\left(\gamma^{\max }-\gamma^{\min }\right)/\varepsilon\right)\right)$, where $I_c$ represents the number of iterations required for Algorithm \ref{Algorithm 1} to converge to the Nash equilibrium, $N$ is the number of C-UAVs, $\gamma^{\min}$ and $\gamma^{\max }$ are the lower and upper bounds of $\gamma$ respectively, and $\varepsilon$ is the search accuracy. 
\end{theorem}
\begin{proof}
{\color{b} The detailed proof is given in Appendix D of the supplemental material.}
\end{proof}

%
%
\section{Simulation Results}
\label{sec:Simulation Results and Analysis}

\par In this section, we perform simulations to validate the effectiveness of our proposed MVTORA approach. Specifically, all the simulations are conducted in MATLAB R2021a on a desktop computer with an AMD Ryzen 7-5800H 3.20-GHz CPU and 16-GB RAM.

%
%
\subsection{Simulation Setup}
\label{subsec:Simulation setups}
\par We consider a three-layer multi-UAV-assisted post-disaster rescue architecture within the area of $2000\times 2000$ $\text{m}^2$, the coordinates of the central point are set as $[0,0,0]$, {\color{b} the distribution density of rescue vehicles is set to $200\ \text{vehicles/km}^2$,} the area is divided equally into the square grids with $400\times 400$ $\text{m}^2$, and $15$ C-UAVs are randomly assigned to perform air search and rescue missions. The flight path of each C-UAV is set to be a circular trajectory with a radius of 100 $\text{m}$ around the center of the square grid, where the C-UAV flies at a constant speed $V = 20$ $\text{m/s}$ and a fixed height $H = 100$ $\text{m}$. {\color{b} In addition, the task generation probability $\rho_n$ is assumed to be uniformly distributed in $[0.8,1]$.} Table \ref{parameters} summarizes the initial values of the other parameters. 

\par To evaluate the performance of the proposed MVTORA approach, we compare it with the following seven benchmark schemes
\begin{itemize}
    \item Entire local computing (ELC): all C-UAVs process their tasks locally.
    \item Entire MEC computing (EMC): all C-UAVs offload their tasks to the E-UAV for execution.
    \item VFC-assisted task offloading (VTO): the tasks generated by C-UAVs can be processed locally or offloaded to ground vehicles for execution. 
    \item MEC-assisted task offloading (MTO): the tasks generated by C-UAVs can be processed locally or offloaded to the E-UAV for execution. 
    \item Only task offloading decision optimization (TODO)~\cite{computationmodel2}: only the task offloading decisions of C-UAVs is optimized, while the edge computation resources are distributed evenly and the vehicle fog nodes are selected randomly with evenly divided tasks.
    \item \textcolor{b}{Non-cooperative game based task offloading (NGTO)~\cite{WangLTZDCZ20}: each C-UAV competitively decides the optimal offloading probability by playing a distributed non-cooperative game.
    \item Dragonfly algorithm (DA)-based task offloading and resource allocation (DATORA)~\cite{Mirjalili16}: the DA is used to solve task offloading and resource allocation.}
\end{itemize}
\begin{table}
	\caption{Simulation Parameters}
	\begin{center}
		\setlength{\tabcolsep}{0.5mm}{
			\begin{tabular}{|c|c|c|c|}
				\hline
				Parameters     &Values                        &Parameters               &Values\\
                    \hline
				$H$            &$100$m                         &$\Delta t$   &$1$s~\cite{parameter-1}\\
                    \hline
                    $H_u$          &$300$m                        &$F_u^{\text{max}}$       &$30$GHz\\
                    \hline
                    $f_m^{\text{veh}}$  &$[0,1]$GHz~\cite{Sun2023}  &$f_n^{\text{uav}}$       &$[1,2]$GHz\\
                    \hline
                    $D_n$      &$[1,3]$Mb~\cite{review-A8}         &$\eta_n$     &$[100,1000]$cycles/bit\\
                    \hline
                    $\Psi$    &$\frac{\pi}{4}$~\cite{parameter-1}    &$T_n^{\text{max}}$       &$[0.5,1]$s\\
                    \hline
                    $\beta_0$      &$1.42\times10^{-4}$~\cite{parameter-1}      &$\kappa$     &$0.2$~\cite{parameter_2}\\
                    \hline
                    $a$    &$10$~\cite{communicationmodel1}    &$b$    &$0.6$~\cite{communicationmodel1}\\
                    \hline
                    $\mu$          &$2.3$~\cite{parameter_2}       &$P_{n,m}$     &$20$dBm\\
                    \hline
                    $P_{n,u}$    &$20$dBm                        &$K_n$                    &$5$\\
                    \hline
                    $\sigma^2$     &$-174$dBm/Hz                  &$B$                      &$200$KHz\\
                    \hline
                    $k$           &$10^{-28}$  &$\rho_0$          &$0.001$\$/GHz\\
                    \hline
				$\alpha_n$    &$0.9$                         &$\beta_n$           &$0.1$\\ 
                    \hline
				$pc$    &$0.8$                            &$pm$           &$0.1$\\ 
                    \hline
				$G$    &$200$                            &$L$           &$50$\\ 
                \hline
		\end{tabular}}
            \label{parameters}
	\end{center}
\end{table}
%
%
\subsection{Evaluation Results}
\label{subsec:Evaluation results}

\par In this section, we first evaluate the convergence and overall system performance of the proposed MVTORA. Furthermore, we compare the impacts of different parameters on the performance of the proposed MVTORA and the benchmark schemes.
%
%
\subsubsection{Convergence and Performance}
\label{subsec:convergence and performance}
\par {\color{b} In this sub-section, we evaluate the performance of the proposed algorithm with different iterations and time slots, respectively. }

\begin{figure}[!hbt] 
	\setlength{\abovecaptionskip}{2pt}%
	\setlength{\belowcaptionskip}{2pt}%
	\centering
	\includegraphics[width =3.5in]{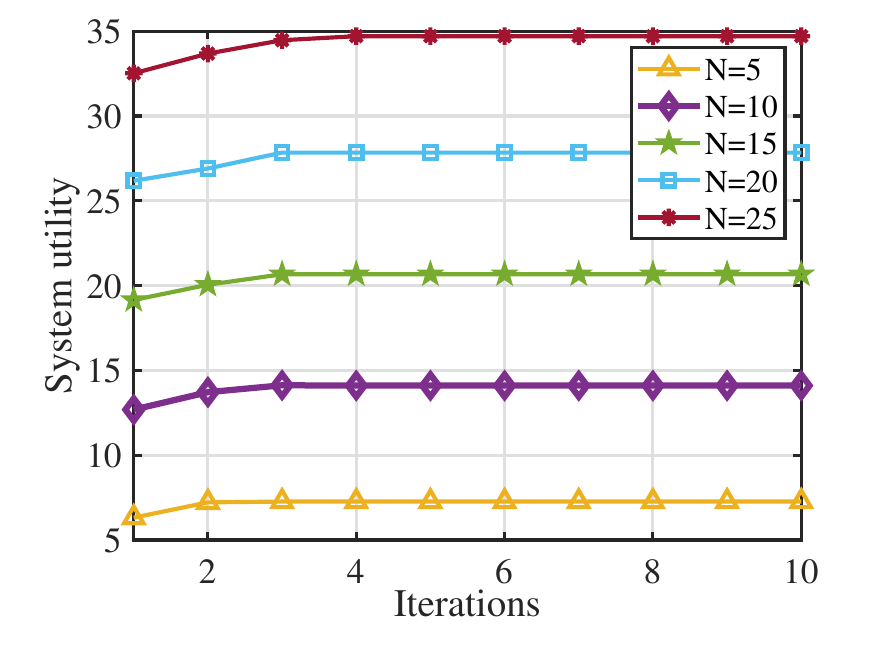}
	\caption{Convergence of MVTORA algorithm.}
    \label{fig.convergence}
\end{figure}
 \begin{figure}[!hbt] 
	\setlength{\abovecaptionskip}{2pt}%
	\setlength{\belowcaptionskip}{2pt}%
	\centering
	\includegraphics[width =3.5in]{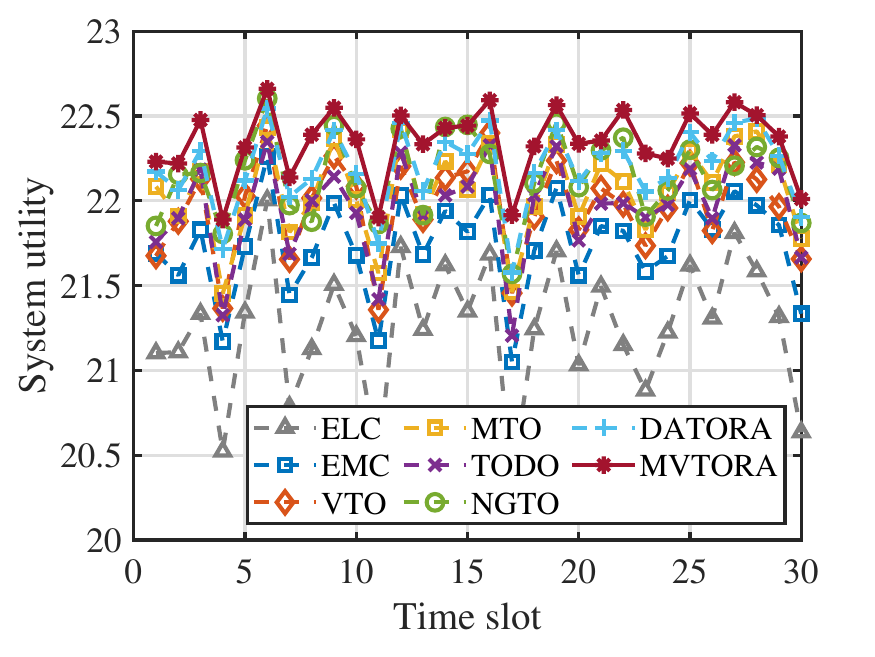}
	\caption{{\color{b} System utility with respect to time slot.}}
    \label{fig.utility-time}
\end{figure}
 
 \par Fig.~\ref{fig.convergence} shows the convergence of the proposed MVTORA under different numbers of C-UAVs. It can be observed that as the number of iterations increases, the system utility keeps on increasing and converges into Nash equilibrium. This is because as the number of iterations increases, each C-UAV attempts to update the offloading strategy to obtain a satisfied utility and eventually reaches the Nash equilibrium state, where no C-UAV can improve the utility even further by unilaterally changing its offloading decision. Thus, the results in the figure demonstrate that the proposed algorithm can achieve a stable state within limited iterations in the scenarios with varying densities. 

{\color{b} \par Fig.~\ref{fig.utility-time} compares the system utility (i.e., the total utility of all C-UAVs) among the abovementioned seven comparison schemes and the proposed MVTORA. Note that in each time slot, real-time decisions of task offloading and resource allocation are determined by running the proposed MVTORA iteratively. Specifically, at the beginning of each time slot, the essential information exchange is performed and our proposed algorithm is executed to make real-time decisions. Our proposed algorithm updates the offloading decisions of C-UAVs iteratively until the offloading decisions of all C-UAVs no longer change, i.e., a Nash equilibrium state is reached. Then, the C-UAVs perform their computing tasks based on the obtained offloading decisions.} 

Moreover, it can be observed from Fig.~\ref{fig.utility-time} that the system utility exhibits irregular fluctuations over time slots. Obviously, this is mainly due to the time-varying nature of the network. Specifically, the computation tasks generated by C-UAVs, the quality of communication links, and the under-utilized computation resources of vehicle fog nodes vary over different time slots. {\color{b} On the other hand, it can be seen that as time elapses, MVTORA outperforms ELC, EMC, VTO, MTO, TODO, NGTO, and DATORA in terms of system utility,} and the reasons are given as follows. First, the entire local computing strategy of ELC could cause overloads of C-UAVs that are restricted by the limited computing resources and onboard battery, leading to additional costs of delay and energy consumption. Furthermore, although the EMC scheme employs the MEC server for task offloading, the limited computing resources of the E-UAV could be a performance bottleneck, especially in disaster areas where terrestrial infrastructures are difficult to be deployed. Moreover, the VFC-assisted offloading strategy of VTO depends heavily on the underutilized resources of ground rescue vehicles, which is insufficient for C-UAVs to perform real-time mission-critical tasks due to the limited resources and mobility of vehicles. {\color{b} Besides, TODO, NGTO, and DATORA outperform ELO, EMC, and VTO, since they integrate the abilities of MEC and VFC. {\color{b} However, TODO only optimizes task offloading without considering the computing resource management, the competitive offloading of NGTO could lead to congestion at the E-UAV, and the heuristic-based strategy of DATORA is sensitive to the initial conditions of the problem, all of which could reduce the accuracy of problem-solving.} Therefore, TODO, NGTO, and DATORA yield worse performances compared to the proposed method. Finally, the superior performance of the proposed MVTORA is mainly due to the joint optimization of the task offloading and aerial-terrestrial resource allocation.} According to the results in Fig.~\ref{fig.utility-time}, it can be concluded that MVTORA has the overall superior performance in the system utility among the eight algorithms. 
%
%
\subsubsection{Impact of Parameters}
\label{subsec:impact of parameters}
\begin{figure*}[!hbt] 
	\centering
	\setlength{\abovecaptionskip}{2pt}%
	\setlength{\belowcaptionskip}{2pt}%
	\subfigure[]
	{
		\begin{minipage}[t]{0.31\linewidth}
			\centering
			\includegraphics[scale=0.4]{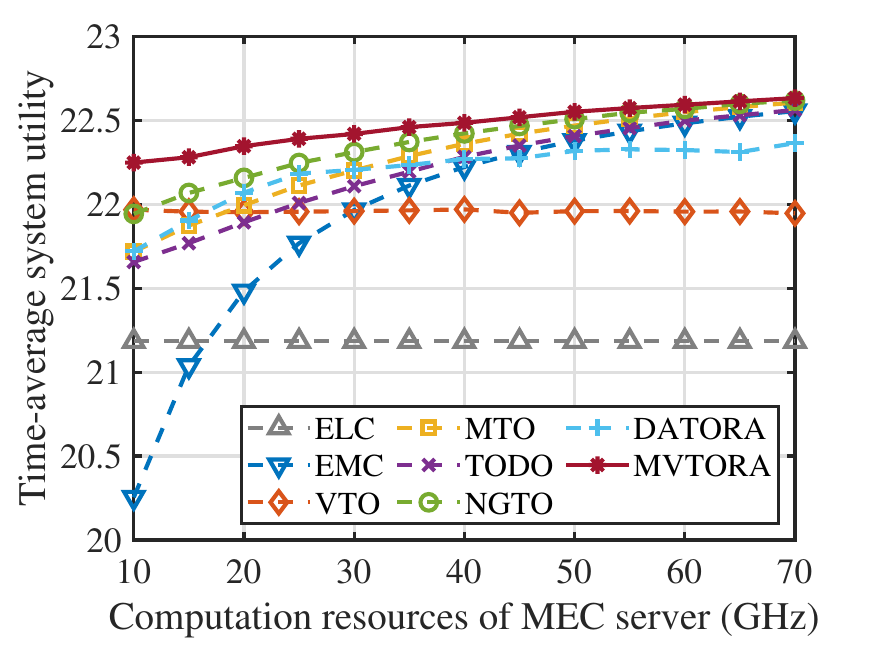}
		\end{minipage}
	}
	\subfigure[]
	{
		\begin{minipage}[t]{0.31\linewidth}
			\centering
			\includegraphics[scale=0.4]{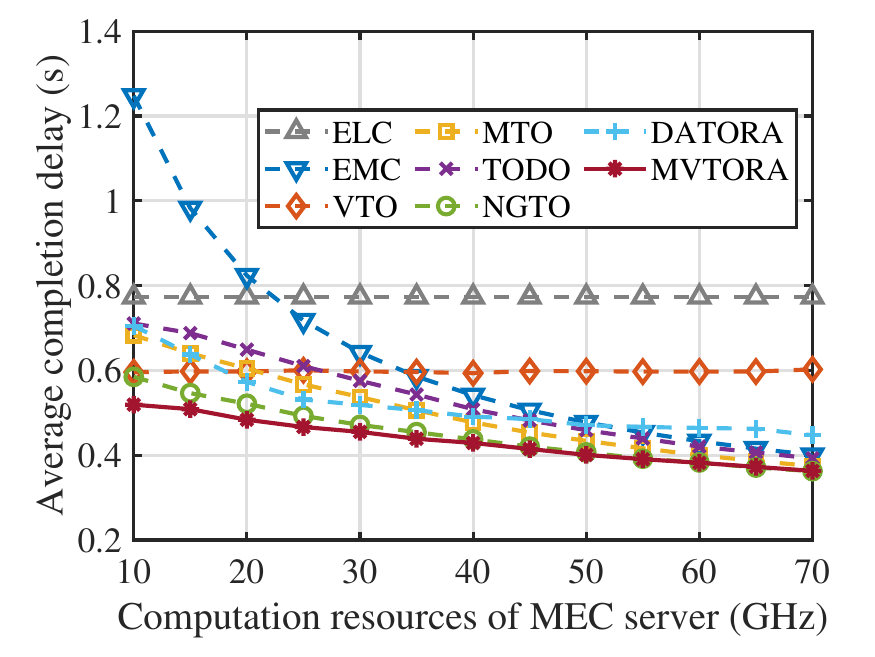}	
		\end{minipage}
	}
	\subfigure[]
	{
		\begin{minipage}[t]{0.31\linewidth}
			\centering
			\includegraphics[scale=0.4]{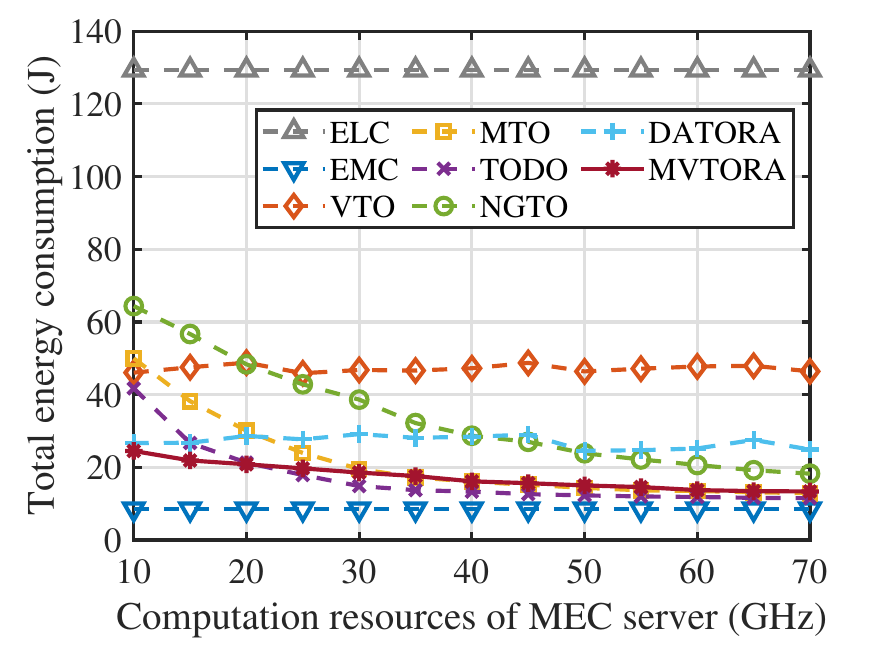}
		\end{minipage}
	}
	\centering
	\caption{{\color{b} System performance with respect to the computation resources of E-UAV. (a) Time-average system utility. (b) Average completion delay. (c) Total energy consumption.}}
	\label{fig_MEC}
	\vspace{-1em}
\end{figure*}
\begin{figure*}[!hbt] 
	\centering
	\setlength{\abovecaptionskip}{2pt}%
	\setlength{\belowcaptionskip}{2pt}%
	\subfigure[]
	{
		\begin{minipage}[t]{0.31\linewidth}
			\centering
			\includegraphics[scale=0.4]{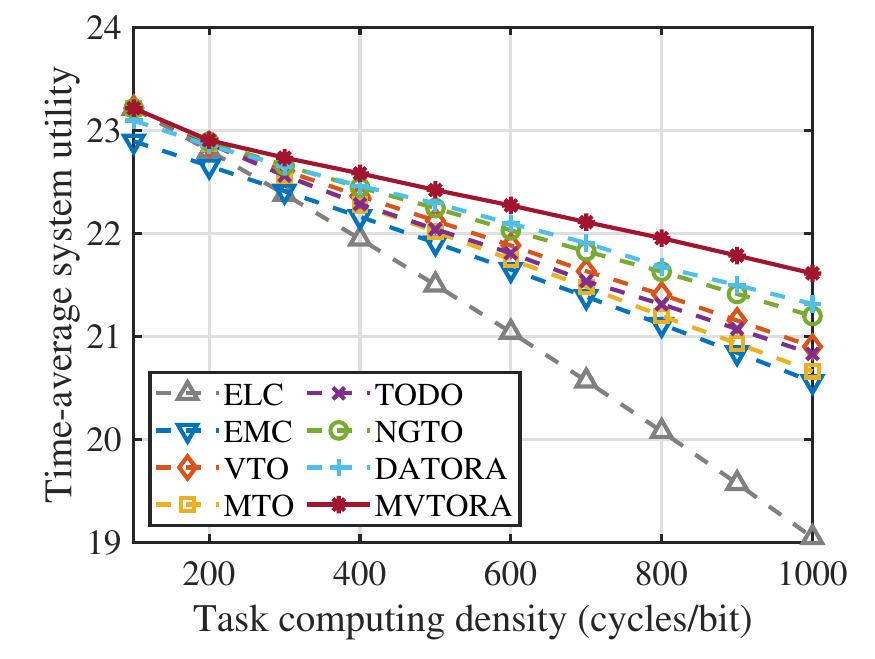}
		\end{minipage}
	}
	\subfigure[]
	{
		\begin{minipage}[t]{0.31\linewidth}
			\centering
			\includegraphics[scale=0.4]{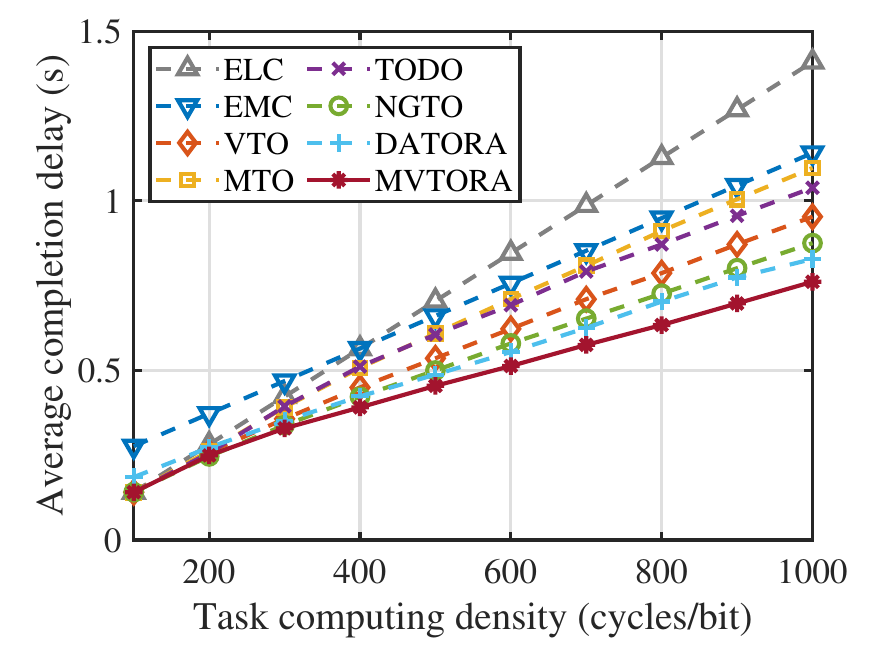}	
		\end{minipage}
	}
	\subfigure[]
	{
		\begin{minipage}[t]{0.31\linewidth}
			\centering
			\includegraphics[scale=0.4]{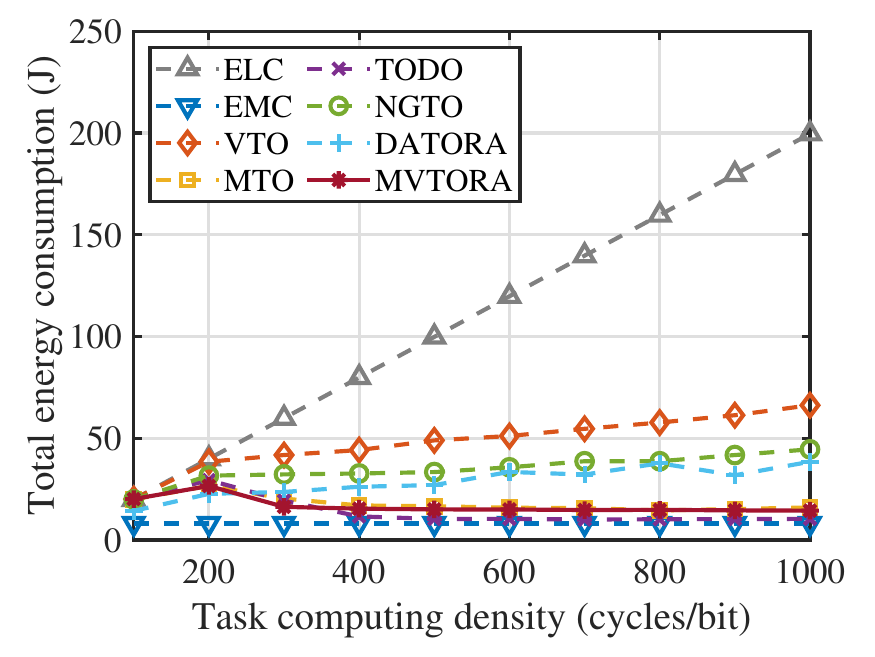}
		\end{minipage}
	}
	\centering
	\caption{{\color{b} System performance with respect to the task computation density of C-UAVs. (a) Time-average system utility. (b) Average completion delay. (c) Total energy consumption.}}
	\label{fig_taskDensity}
	\vspace{-1em}
\end{figure*}
\begin{figure*}[!hbt] 
	\centering
	\setlength{\abovecaptionskip}{2pt}%
	\setlength{\belowcaptionskip}{2pt}%
	\subfigure[]
	{
		\begin{minipage}[t]{0.31\linewidth}
			\centering
			\includegraphics[scale=0.4]{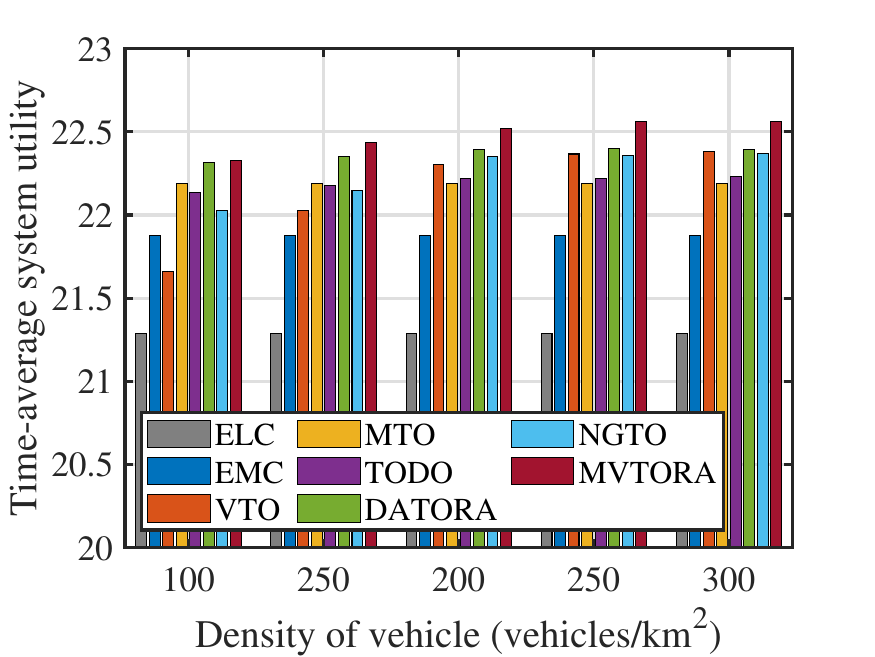}
		\end{minipage}
	}
	\subfigure[]
	{
		\begin{minipage}[t]{0.31\linewidth}
			\centering
			\includegraphics[scale=0.4]{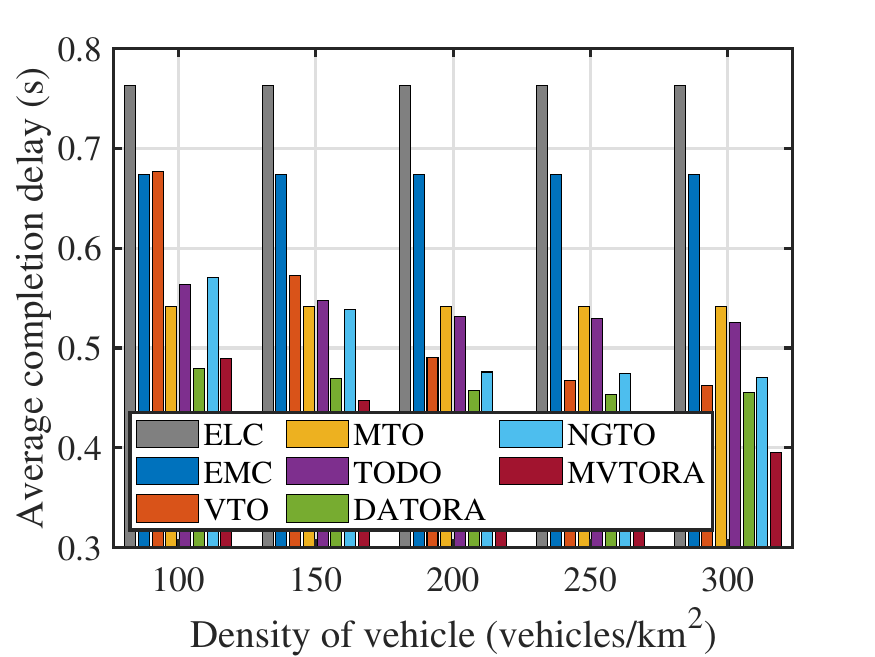}	
		\end{minipage}
	}
	\subfigure[]
	{
		\begin{minipage}[t]{0.31\linewidth}
			\centering
			\includegraphics[scale=0.4]{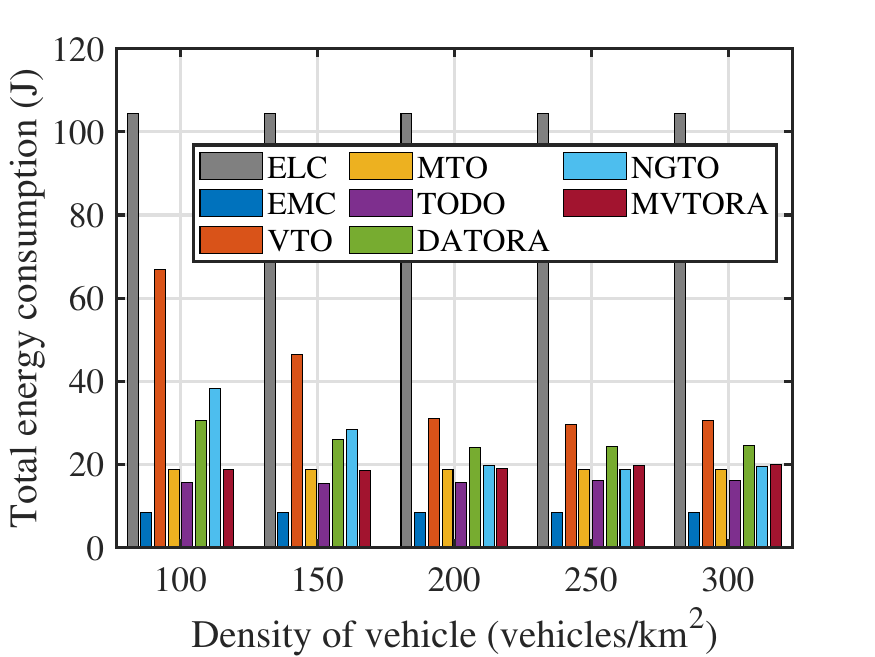}
		\end{minipage}
	}
	\centering
	\caption{{\color{b} System performance with respect to the density of vehicles. (a) Time-average system utility. (b) Average completion delay. (c) Total energy consumption.}}
	\label{fig_vehDensity}
	\vspace{-1em}
\end{figure*}

\par \textbf{\textit{Impact of Edge Computation Resources.}} {\color{b} Figs.~\ref{fig_MEC}(a), \ref{fig_MEC}(b) and \ref{fig_MEC}(c) show the impact of the computation resources of E-UAV on time-average system utility, average completing delay, and total energy consumption for the eight approaches.} First, ELC and VTO maintain nearly constant performance in terms of time-average system utility, average completing delay, and total energy consumption regardless of variations in edge computation resources. This is obviously because ELC and VTO do not use the edge computing resources as explained in Section~\ref{subsec:Simulation setups}. {\color{b} Moreover, with the increasing of edge computation resources, the time-average system utilities of MTO, TODO, NGTO, and DATORA show increasing trends, while their average completing delay and total energy consumption show the opposite trends. This is because more tasks can be offloaded to the edge server with the increasing of the edge computing resources}. Additionally, as the edge computing resources increase, EMC shows the most drastic variation trends in terms of time-average system utility and average completing delay while keeping invariant in terms of total energy consumption. The main reasons are as follows. First, EMC is completely dependent on edge computing, which makes it sensitive to edge computation resources. Second, the total energy consumption is induced by the energy consumption of task transmission, which is independent of the edge computing resources. Finally, compared to the benchmark algorithms, the time-average system utility, average completing delay, and total energy consumption of the proposed MVTORA exhibit slightly steady variation trends with respect to the edge computation resources. The reason is that MVTORA integrates aerial and terrestrial edge computing resources to alleviate the overload of C-UAVs. In conclusion, the set of results in Fig.~\ref{fig_MEC} demonstrates that the proposed MVTORA is able to achieve sustainable computation resource utilization and effective performances in terms of time-average system utility, average completing delay, and total energy consumption with varying computing resources of the edge server.

\par \textbf{\textit{Impact of Task Computation Density.}} {\color{b} Figs.~\ref{fig_taskDensity}(a), \ref{fig_taskDensity}(b), and \ref{fig_taskDensity}(c) display the impact of the task computing density of C-UAVs on the time-average system utility, the average completing delay, and the total energy consumption among the eight approaches, respectively.} First, EMC shows inferior performance in the time-average system utility and average completing delay compared to the other schemes except for ELC, and exhibits the lowest total energy consumption among the six schemes. This is mainly because the MEC offloading strategy of EMC produces lower transmission energy consumption, but the aggregated workload could overload the edge server as the computing density increases. {\color{b} Furthermore, it can be observed that ELC, MTO, VTO, TODO, NGTO, DATORA, and MVTORA achieve similar performances on time-average system utility, average completing delay, and total energy consumption when the task computing density is relatively small (less than 300 cycles/bit). The reason is that local computing could be a favorable choice for these schemes when the computing density of the task is small, which does not generate additional costs of transmission delay and energy consumption. Besides, MTO, TODO, NGTO, DATORA, and MVTORA reveal similar trends in total energy consumption, with a slight initial increase followed by a decrease, and then remaining approximately constant as task computing density increases. This is mainly because most tasks are offloaded to the E-UAVs for execution as the task computing density further increases. As a result, the total energy consumption is mainly form the task transmission energy consumption, which leads to a decreasing trend.} Finally, the MVTORA algorithm shows superior performance in the time-average system utility and average completing delay among the six schemes. This is because the MVTORA algorithm is empowered by aerial and terrestrial edge capabilities by integrating MEC and VFC capabilities, which could deal with high-complexity tasks through parallel processing. The set of simulation results indicates the proposed MVTORA is able to adapt to varying computing densities with relatively superior performances in system utility, completion delay, and energy consumption, especially {\color{color1} in the heavy workload scenario}.

\par \textbf{\textit{Impact of Vehicle Distribution Density.}} Figs.~\ref{fig_vehDensity}(a),~\ref{fig_vehDensity}(b), and~\ref{fig_vehDensity}(c) describe the impact of the vehicle distribution density on the time-average system utility, the average completing delay, and the total energy consumption among the comparative algorithms. First, the time-average system utility, average completing delay, and total energy consumption of ELC, EMC, and MTO remain constant with the increasing of vehicle distribution density. Obviously, this is because ELC, EMC, and MTO do not exploit terrestrial edge computing resources of rescue vehicles. Moreover, TODO shows a slow and slight variation trend in the time-average system utility, the average completing delay, and the total energy consumption with increasing vehicle distribution density. This is mainly due to the lack of an effective VFC-assisted task offloading algorithm, which could lead to insufficient utilization of vehicle fog node resources. Furthermore, as the vehicle distribution density increases, VTO exhibits a significant upward trend in the time-average system utility and substantial downward trends in the average completing delay and total energy consumption. The reason is that VTO can effectively utilize the resources of vehicle fog nodes by optimizing VFC-assisted task offloading. 
{\color{b} In addition, DATORA shows non-obvious trends in terms of the time-average system utility and average completing delay. This is because the heuristic-based method of DATORA relies on the initial condition and simplification of the original problem, which can lead to inaccurate or suboptimal results of resource allocation for vehicle fog nodes, and therefore cannot fully utilize the resources of vehicle fog nodes.} Finally, it can be observed that MVTORA exhibits significantly superior performance compared to the other algorithms in terms of the time-average system utility and average completion delay with low energy consumption. The set of simulation results indicates that the proposed algorithm has better scalability with the increasing of vehicle distribution density.
%
%
{\color{b}
\section{Discussion}
\label{sec:Discussion}

\par {\color{b} In this section, we discuss the generalizability of our method with regard to specific vehicle distribution and mobility models, as well as the rationale behind the selected offloading strategy.}
%
%
\subsection{Impact of Vehicle Distribution and Mobility}
\label{subsec:veh-dis-mob}

\par \textcolor{b} {To explore the generalizability of our approach, we verify the effectiveness of our proposed approach for different vehicle distributions and mobility models. Specifically, we consider the following three cases: \textbf{\textit{i)}} random distribution and random walk model (RD-RWM)~\cite{pearson1905}, \textbf{\textit{ii)}} mobile traffic model (MTM)~\cite{review-A8}, and \textbf{\textit{iii)}} Poisson cluster process and Markovian way point model (PCP-MWPM)~\cite{TabassumSH19}, and the corresponding simulation results are shown in Fig. 1 of Appendix E in the supplemental material. The simulation results illustrate that our proposed approach is also applicable to other vehicle distribution and mobility models.}

%
%
\subsection{Comparison with Mixed Task Offloading Scheme}
\label{subsec:comp-MTOS}

\par \textcolor{b} {To demonstrate the rationality of the offloading decision for our proposed approach,} we compare the proposed task offloading scheme with the mixed task offloading scheme of local computing, MEC, and VFC. Specifically, in Appendix F of the supplementary martial, we first analyze the limitations of the mixed offloading scheme. Then, we compare the performance of our method and the mixed task offloading scheme in terms of time-average system utility, average task completion delay, total energy consumption, and average algorithm running time. The analysis and simulation results demonstrate that our proposed approach is more suitable for the considered post-disaster rescue scenarios.
}
%
%
\section{Conclusion}
\label{sec:Conclusion}

\par In this work, we study the task offloading and resource allocation in UAV networks for post-disaster rescue. First, by integrating the aerial and terrestrial computing capabilities, we propose an MEC-VFC-assisted three-layer computing architecture for post-disaster rescue, which consists of a vehicular fog layer, a UAV edge layer, and a UAV client layer. Furthermore, the JTRAOP is formulated to maximize the time-average utility of the system by jointly optimizing task offloading and computing resource allocation. Since the problem is NP-hard, we develop an MVTORA approach with low complexity to separate the initial problem into the components of task offloading and resource allocation, which are solved by proposing
a game theory-based algorithm for task offloading decision, a convex optimization-based algorithm for MEC resource allocation, and an evolutionary computation-based hybrid algorithm for VFC resource allocation. Simulation results demonstrate the superiority of the proposed MVTORA approach in terms of time-average system utility, average task completion delay, and total energy consumption. In the future, our work will be extended to include UAV trajectory optimization.
\bibliographystyle{IEEEtran}
\bibliography{myref}

\begin{IEEEbiography}[{\includegraphics[width=1in,height=1.25in,clip,keepaspectratio]{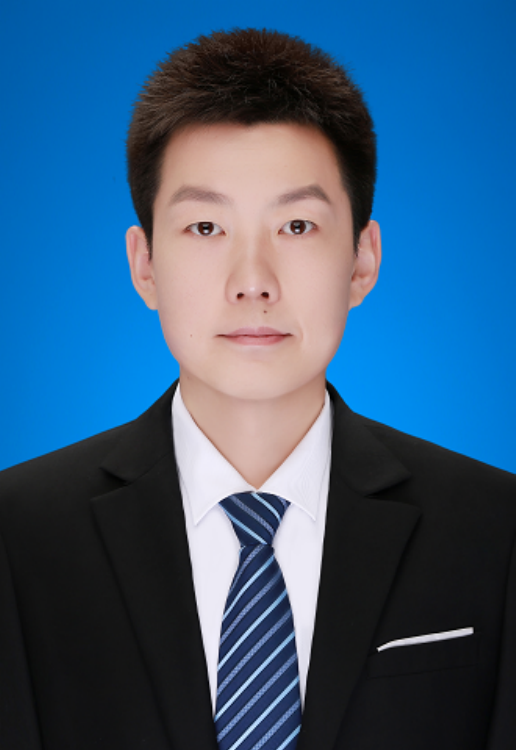}}]{Geng Sun} (S'17-M'19) received the B.S. degree in communication engineering from Dalian Polytechnic University, and the Ph.D. degree in computer science and technology from Jilin University, in 2011 and 2018, respectively. He was a Visiting Researcher with the School of Electrical and Computer Engineering, Georgia Institute of Technology, USA. He is an Associate Professor in College of Computer Science and Technology at Jilin University, and his research interests include wireless networks, UAV communications, collaborative beamforming and optimizations.
\end{IEEEbiography}

\begin{IEEEbiography}[{\includegraphics[width=1in,height=1.25in,clip,keepaspectratio]{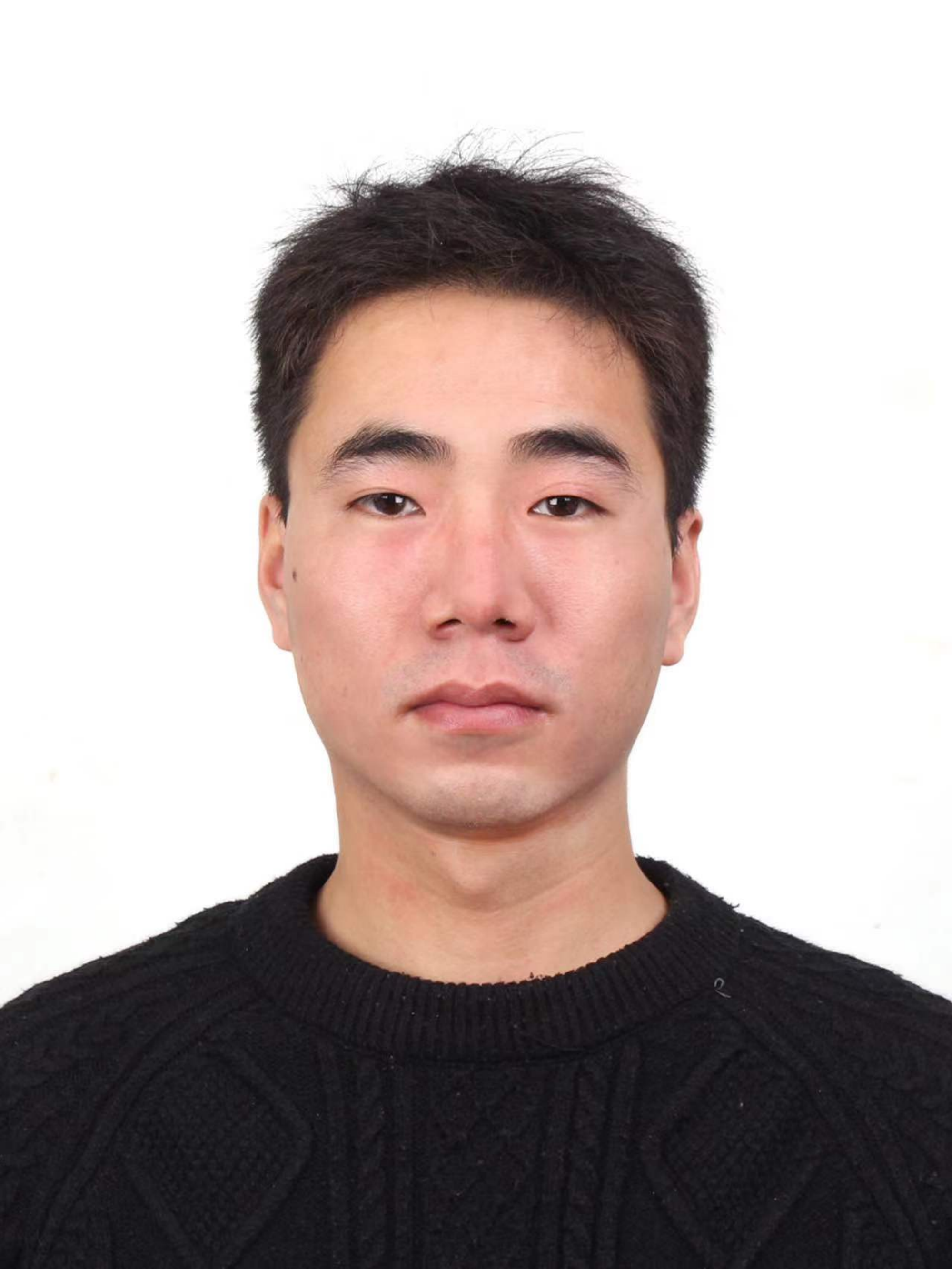}}]{Long He} received a BS degree in Computer Science and Technology from Chengdu University of Technology, Sichuan, China, in 2019. He is currently working toward the PhD degree in Computer Science and Technology at Jilin University, Changchun, China. His research interests include vehicular networks and edge computing.
\end{IEEEbiography}

\begin{IEEEbiography}[{\includegraphics[width=1in,height=1.25in,clip,keepaspectratio]{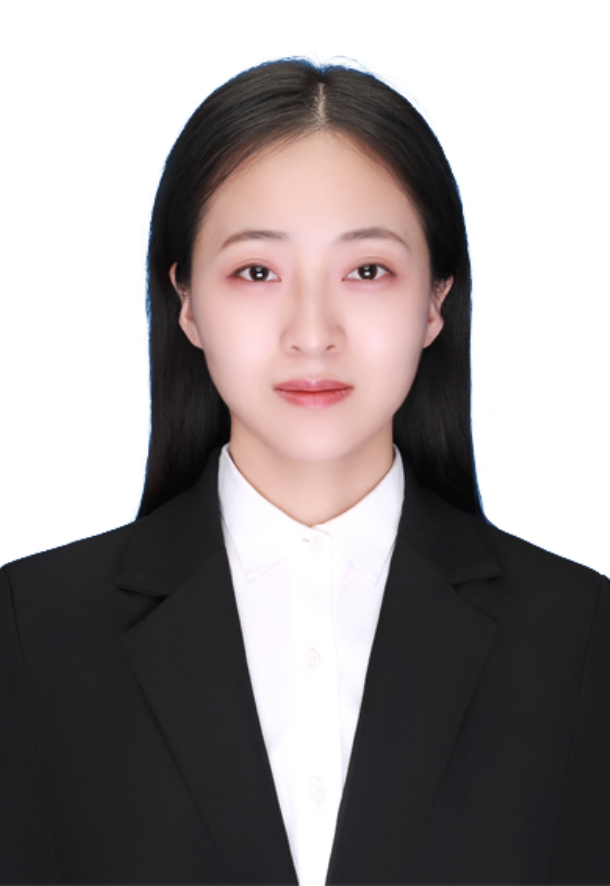}}]{Zemin Sun} received a BS degree in Software Engineering, an MS degree and a Ph.D degree in Computer Science and Technology from Jilin University, Changchun, China, in 2015, 2018, and 2022, respectively. Her research interests include vehicular networks, edge computing, and game theory. 
\end{IEEEbiography}

\begin{IEEEbiography}[{\includegraphics[width=1in,height=1.25in,clip,keepaspectratio]{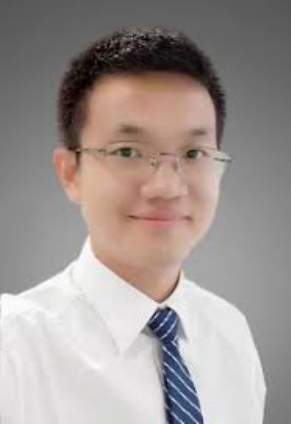}}]{Qingqing Wu} (S’13-M’16-SM’21) received the B.Eng. and the Ph.D. degrees in Electronic Engineering from South China University of Technology and Shanghai Jiao Tong University (SJTU) in 2012 and 2016, respectively. From 2016 to 2020, he was a Research Fellow in the Department of Electrical and Computer Engineering at National University of Singapore. He is currently an Associate Professor with Shanghai Jiao Tong University. His current research interest includes intelligent reflecting surface (IRS), unmanned aerial vehicle (UAV) communications, and MIMO transceiver design. He has coauthored more than 100 IEEE journal papers with 26 ESI highly cited papers and 8 ESI hot papers, which have received more than 18,000 Google citations. He was listed as the Clarivate ESI Highly Cited Researcher in 2022 and 2021, the Most Influential Scholar Award in AI-2000 by Aminer in 2021 and World’s Top 2\% Scientist by Stanford University in 2020 and 2021.
\end{IEEEbiography}

\begin{IEEEbiography}[{\includegraphics[width=1in,height=1.25in,clip,keepaspectratio]{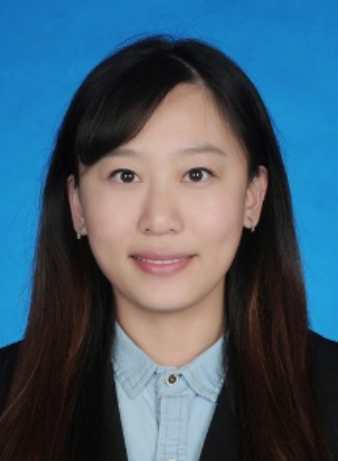}}]{Shuang Liang} received the B.S. degree in Communication Engineering from Dalian Polytechnic University, China in 2011, the M.S. degree in Software Engineering from Jilin University, China in 2017, and the Ph.D. degree in Computer Science from Jilin University, China in 2022. She is a post-doctoral in the School of Information Science and Technology, Northeast Normal University, and her research interests focus on wireless communication and UAV networks.
\end{IEEEbiography}

\begin{IEEEbiography}[{\includegraphics[width=1in,height=1.25in,clip,keepaspectratio]{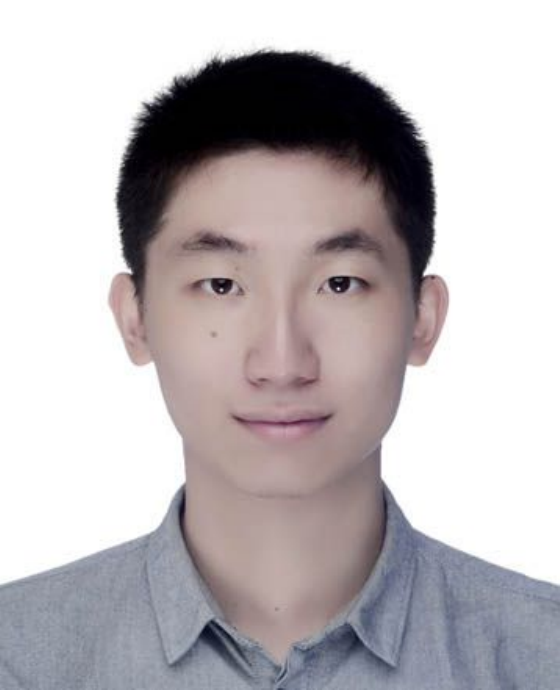}}]{Jiahui Li}
(S'21) received a BS degree in Software Engineering, and an MS degree in Computer Science and Technology from Jilin University, Changchun, China, in 2018 and 2021, respectively. He is currently studying Computer Science at Jilin University to get a Ph.D. degree, and also a visiting Ph. D. at Singapore University of Technology and Design (SUTD), Singapore. His current research focuses on UAV networks, antenna arrays, and optimization.
\end{IEEEbiography}

\begin{IEEEbiography}[{\includegraphics[width=1in,height=1.25in,clip,keepaspectratio]{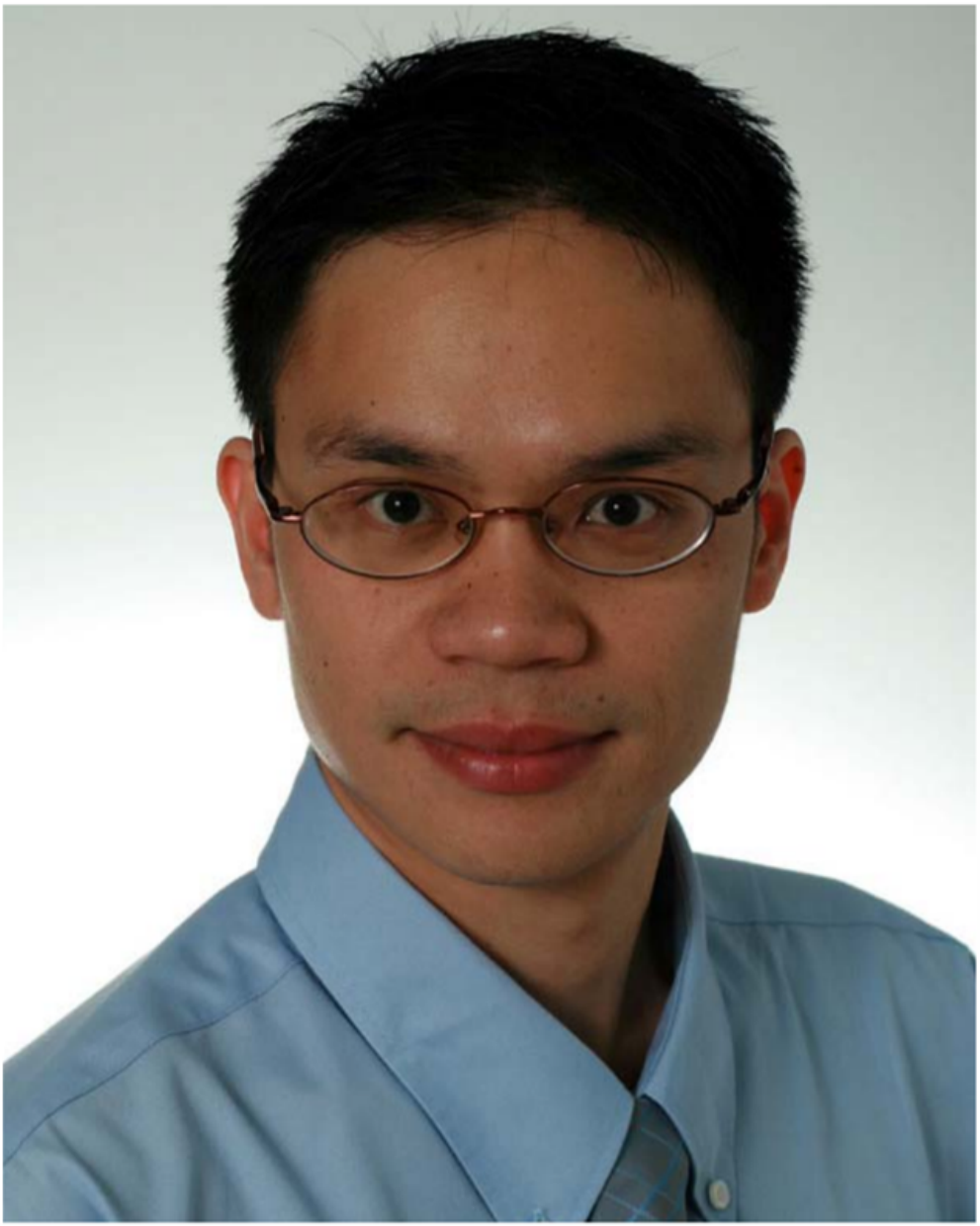}}]{Dusit Niyato} (Fellow, IEEE) received the B.Eng. degree from the King Mongkuts Institute of Technology Ladkrabang (KMITL), Thailand, in 1999, and the Ph.D. degree in electrical and computer engineering from the University of Manitoba, Canada, in 2008. He is currently a Professor with the School of Computer Science and Engineering, Nanyang Technological University, Singapore. His research interests include the Internet of Things (IoT), machine learning, and incentive mechanism design.
\end{IEEEbiography}

\begin{IEEEbiography}[{\includegraphics[width=1in,height=1.25in,clip,keepaspectratio]{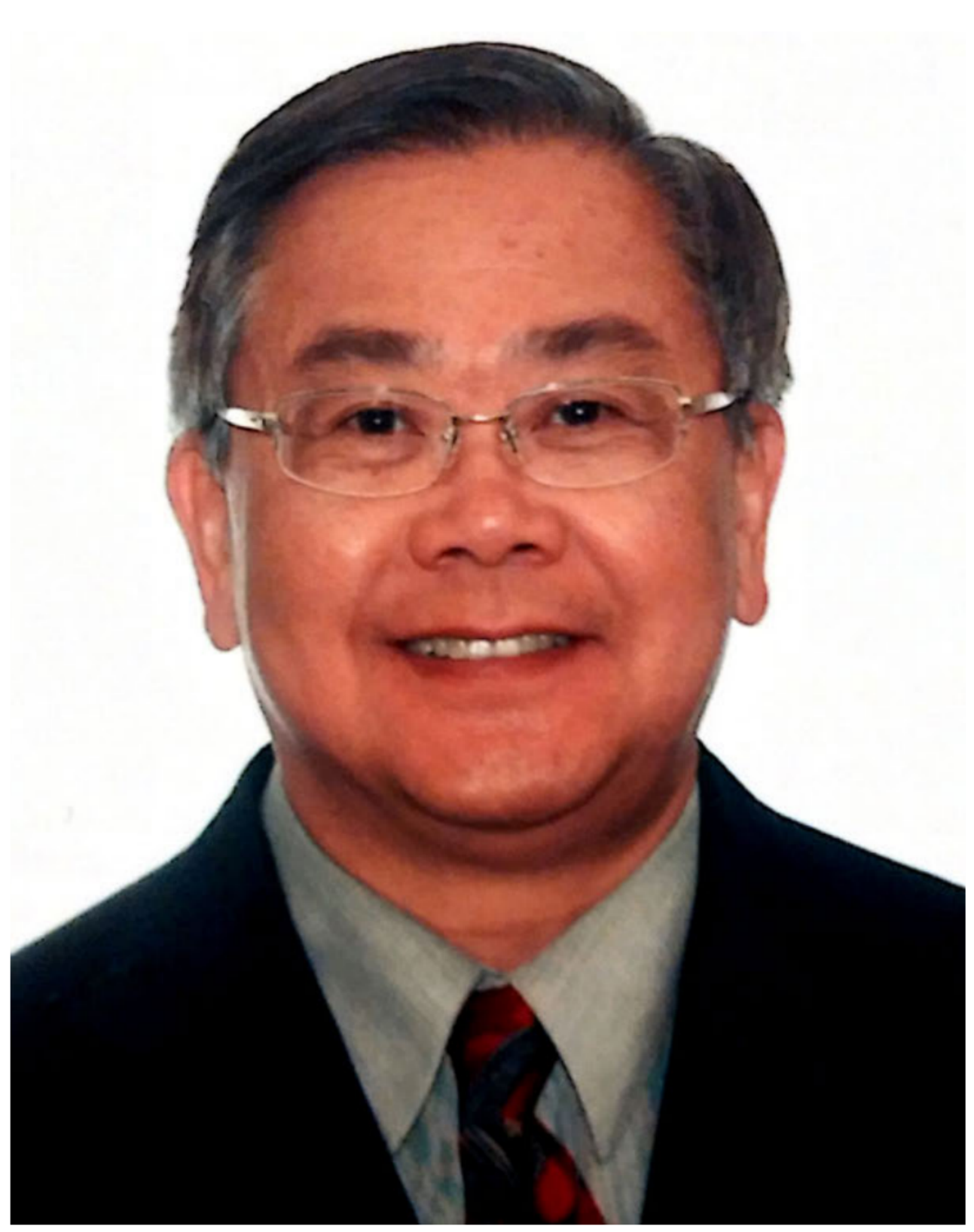}}]{Victor C. M. Leung} (Life Fellow, IEEE) is a Distinguished Professor of computer science and software engineering with Shenzhen University,
China. He is also an Emeritus Professor of electrial and computer engineering and the Director of the Laboratory for Wireless Networks and Mobile Systems at the University of British Columbia (UBC). His research is in the broad areas of wireless networks and mobile systems. He has co-authored more than 1300 journal/conference papers and book chapters. Dr. Leung is serving on the editorial boards of IEEE Transactions on Green Communications and Networking, IEEE Transactions on Cloud Computing, IEEE Access, and several other journals. He received the IEEE Vancouver Section Centennial Award, 2011 UBC Killam Research Prize, 2017 Canadian Award for Telecommunications Research, and 2018 IEEE TCGCC Distinguished Technical Achievement Recognition Award. He co-authored papers that won the 2017 IEEE ComSoc Fred W. Ellersick Prize, 2017 IEEE Systems Journal Best Paper Award, 2018 IEEE CSIM Best Journal Paper Award, and 2019 IEEE TCGCC Best Journal Paper Award. He is a Life Fellow of IEEE, and a Fellow of the Royal Society of Canada, Canadian Academy of Engineering, and Engineering Institute of Canada. He is named in the current Clarivate Analytics list of Highly Cited Researchers.
\end{IEEEbiography}
\end{document}